\documentclass[11pt]{article}
\usepackage[english]{babel}			 
\usepackage[utf8]{inputenc}		
\usepackage[T1]{fontenc}										
\usepackage{amsmath,amsfonts,amssymb,amsthm,cancel,
calculator,calc,mathtools,empheq,latexsym}
\usepackage{booktabs,multicol,multirow,tabularx,array}    

\usepackage[colorlinks = true,
            linkcolor = blue,
            citecolor = blue,
            anchorcolor = blue]{hyperref}
\newtheorem{definition}{Definition}[section]
\newtheorem{theorem}{Theorem}[section]

\newtheorem{lemma}[theorem]{Lemma}

\usepackage{enumitem}
\usepackage{enumerate,cite,url}
\usepackage{fullpage}
\usepackage{xcolor}
\usepackage[ruled,vlined,linesnumbered ]{algorithm2e}
\usepackage{comment}

\newcommand{\samp}{\textsf{SAMP}~}
\newcommand{\cond}{\textsf{COND}~}
\newcommand{\pcond}{\textsf{PCOND}~}
\DeclareMathOperator{\coll}{coll}
\newcommand{\countmin}{\textsf{CountMin}~}
\newcommand{\accept}{\textsf{Accept}~}
\newcommand{\reject}{\textsf{Reject}~}
\newcommand{\compare}{\emph{Compare}~}

\title{Testing properties of distributions in the streaming model }
\author{ Sampriti Roy, Yadu Vasudev}

\begin{document}
\maketitle

\begin{abstract}

We study distribution testing in the standard access model and the conditional access model when the memory available to the testing algorithm is bounded. In both scenarios, the samples appear in an online fashion and the goal is to test the properties of distribution using an optimal number of samples subject to a memory constraint on how many samples can be stored at a given time. First, we provide a trade-off between the sample complexity and the space complexity for testing identity when the samples are drawn according to the conditional access oracle. We then show that we can learn a succinct representation of a monotone distribution efficiently with a memory constraint on the number of samples that are stored that is almost optimal. We also show that the algorithm for monotone distributions can be extended to a larger class of decomposable distributions.
\end{abstract}

\section{Introduction}
Sublinear algorithms that analyze massive amounts of data are crucial in many
applications currently. Understanding the underlying probability distribution
that generates the data is important in this regard. In the field of
distribution testing, a sub-field of property testing, the goal is to test
whether a given unknown distribution has a property $\mathcal{P}$ or is far from
having the property $\mathcal{P}$ (where the farness is defined with respect to
total variation distance). Starting from the work of Goldreich and Ron
(\cite{GR00}), a vast literature of work has studied the problem of testing
probability distributions for important properties like identity, closeness,
support size as well as properties relating to the structure of the distribution
like monotonicity, k-modality, and histograms among many others; see Canonne's
survey (\cite{Canonne22}) for an overview of the problems and results.

In the works of Canonne et al (\cite{CRS15}) and Chakraborty et al
(\cite{CFGM13}), distribution testing with conditional samples was studied. In
this model, the algorithm can choose a subset of the support, and the samples of
the distribution conditioned on this subset are generated. This allows adaptive
sampling from the distribution and can give better sample complexity for a
number of problems. In particular, (\cite{CRS15}) and (\cite{CFGM13}) give testers
for uniformity and other problems that use only a constant number of samples.

The natural complexity measure of interest is the number of samples of the
underlying distribution that is necessary to test the property. In many cases,
when data is large, it might be infeasible to store all the samples that are
generated. A recent line of work has been to study the trade-off between the
sample complexity and the space complexity of algorithms for learning and
testing properties of distributions. This model can be equivalently thought of
as a data stream of i.i.d samples from an unknown distribution, with the
constraint that you are allowed to store only a small subset of these samples at
any point in time.

In this work, we study distribution testing problems in the
standard model, and when the algorithm is allowed to condition on sets to better understand
the trade-off between the sample complexity and size. In particular, we study
identity testing and testing whether the unknown distribution is monotone. Our
work borrows ideas from the recent work of Diakonikolas et al (\cite{DGKR19})
and extends the ideas to these problems.

\subsection{Related work}

Testing and estimating the properties of discrete distributions is well-studied in property testing; see (\cite{Canonne22}) for a nice survey of recent results. In our work, we study property testing of discrete distributions under additional memory constraints wherein the algorithm does not have the resources to store all the samples that it obtains.

This line of work has received a lot of attention in recent times. Chien et al (\cite{CLM10}) propose a sample-space trade-off for testing any $(\epsilon,\delta)$-weakly continuous properties, as defined by Valiant (\cite{Valiant}). Another work by Diakonikolas et al (\cite{DGKR19}) studies the uniformity, identity, and closeness testing problems and presents trade-offs between the sample complexity and the space complexity of the tester. They use the idea of a \emph{bipartite collision tester} where instead of storing all the samples in the memory, the testing can be done by storing a subset of samples and counting the collisions between the stored set and the samples that come later. Another line of work (\cite{ABIS19, AMNW22}) focuses on the task of estimating the entropy of distributions from samples in the streaming model, where space is limited. In particular, (\cite{ABIS19}) estimate the entropy of an unknown distribution $D$ up to $\pm \epsilon$ using constant space. Berg et al (\cite{BOS22}) study the uniformity testing problem in a slightly different model where the testing algorithm is modeled as a finite-state machine.


Property testing with memory constraints has also been studied in the setting of streaming algorithms as well. Streaming algorithms were first studied in a unified way starting from the seminal work of Alon et al (\cite{AMS99}) where the authors studied the problem of estimating frequency moments. There is a vast amount of literature available on streaming algorithms (see \cite{muthukrishnan05,McGregor14}). Bathie et al (\cite{BGST21}) have studied property testing in the streaming model for testing regular languages.  Czumaj et al (\cite{CFPS20}) show that every constant-query testable property on bounded-degree graphs can be tested on a random-order stream with constant space. Since this line of work is not directly relevant to our work in this paper, we will not delve deeper into it here.

\subsection{Our results}

In this work, we study the trade-off between sample complexity and space
complexity in both the standard access model and the conditional access
model. In the standard access model, a set of samples can be drawn independently
from an unknown distribution. In the case of the conditional access model, a subset
of the domain is given and samples can be drawn from an unknown distribution
conditioned on the given set. This is similar to a streaming algorithm where the samples are presented to the algorithm, and the algorithm has a memory constraint of $m$ bits; i.e.,
only up to $m$ bits of samples can be stored in memory.

In the standard access model, which we will refer to as \samp, we have a
distribution $D$ over the support $\{1,2,\ldots,n\}$ and the element $i$ is
sampled with probability $D(i)$. In the conditional access model, which we will refer to as \cond, the algorithm can choose a set $S\subseteq \{1,2,\ldots,n\}$ and will obtain samples from the conditional distribution over the set. I.e.\ the sample $i \in S$ is returned with probability $D(i)/D(S)$. In this work, we will work with the case when the conditioning is done on sets of size at most two - we will refer to this conditional oracle as \pcond (\cite{CRS15}). 

Our results are stated below.
\begin{itemize}
\item We propose a memory-efficient identity testing algorithm in the \pcond model when the algorithm is restricted by the memory available to store the samples. We adapt the algorithm of Canonne et al (\cite{CRS15})
  and reduce the memory requirement
  by using the \countmin sketch (\cite{CM04}) for storing the frequencies of the
  samples. The identity testing algorithm uses
  $O(\log^2{n}\log{\log{n}}/m\epsilon^2)$ samples from standard access model where $\frac{\log{n}\sqrt{\log{\log{n}}}}{{\epsilon}}\leq m\leq \frac{\log^2{n}}{\epsilon}$ and an $\tilde{O}(\log^4{n}/\epsilon^4)$ samples from conditional access model and does the following, if $D=D^*$, it returns \accept with probability at least $2/3$, and if $d_{TV}(D, D^*)\geq \epsilon$, it returns \reject with probability at least
  $2/3$. It uses only $O(\frac{m}{\epsilon})$ bits of memory.

    We also observe that by applying oblivious decomposition \cite{Birge},  performing identity and closeness testing on monotone distributions over $[n]$ can be reduced to performing the corresponding tasks on arbitrary distributions over $[O(\log{(n\epsilon)}/\epsilon)]$. We use the streaming model based identity tester from (\cite{DGKR19}) and obtain an $O(\log{(n\epsilon)}\log{\log{(n\epsilon)}}/m\epsilon^5)$ standard access query identity tester for monotone distributions where $\log{\log{(n\epsilon)}}/\epsilon^2\leq m\leq (\log{(n\epsilon)/\epsilon})^{9/10}$. Their closeness testing algorithm also implies a closeness tester for monotone distributions which uses $O(\log{(n\epsilon)} \sqrt{\log{\log{(n\epsilon)}}}/\sqrt{m}\epsilon^3)$ samples from standard access model, where $\log{\log({n}\epsilon)}\leq m\leq \tilde{\Theta}(min(\log{(n\epsilon)}/\epsilon,\log^{2/3}{(n\epsilon)}/\epsilon^{2}))$. Both testers require $m$ bits of memory.
    \item We adapt the idea of the \emph{bipartite collision tester} (\cite{DGKR19}) and give an algorithm that uses $O(\frac{n\log{n}}{m\epsilon^{8}})$ samples from \samp \ and tests if the distribution is monotone or far from being monotone. This algorithm requires only $O(m)$ bits of memory for $\log^2{n}/\epsilon^6\leq m\leq \sqrt{n}/\epsilon^3$. This upper bound is nearly tight since we observe that the lower bound for uniformity testing proved by Diakonikolas et al (\cite{DGKR19}) applies to our setting as well. In particular, we show that the "no" distribution that is used in \cite{DGKR19} is actually far from monotone, and hence the lower bound directly applies in our setting as well.

    \item We extend the idea of the previous algorithm for learning and testing a more general class of distribution called $(\gamma, L)$-decomposable distribution, which includes monotone and $k$-modal distributions. Our algorithm takes $O(\frac{nL\log{(1/\epsilon)}}{m\epsilon^9})$ samples from $D$ and needs $O(m)$ bits of memory where $\log{n}/\epsilon^4 \leq m\leq O(\sqrt{n\log{n}}/\epsilon^3)$.
    
\end{itemize}

\section{Notation and Preliminaries}

Throughout this paper, we study distributions $D$ that are supported over the set $\{1,2,\ldots,n\} = [n]$.
The notion of distance between distributions will be \emph{total variation distance} or \emph{statistical distance} which is defined as follows: for two distributions $D_1$ and $D_2$, the total variation distance, denoted by $d_{TV}(D_1,D_2)=\frac{1}{2} |D_1-D_2|_1=\frac{1}{2} \sum_{i\in [n]} |D_1(x)-D_2(x)|=max_{S\subseteq [n]} ((D_1(S)-D_2(S))$. We will use $\mathcal{U}$ to denote the uniform distribution over $[n]$. We use $|.|_1$ for the $\ell_1$ norm, $||.||_2$ for the $\ell_2$ norm. 

Let $D_1$ and $D_2$ be two distributions over $[n]$, if $d_{TV}(D_1,D_2)\leq \epsilon$, for some $0\leq \epsilon\leq 1$, we say that $D_1$ is $\epsilon$ close to $D_2$. 
Let $\mathcal{D}$ be the set of all probability distributions supported on $[n]$. A property $\mathcal{P}$ is a subset of $\mathcal{D}$. We say that a distribution $D$ is $\epsilon$ far from $\mathcal{P}$, if $D$ is $\epsilon$ far from all the distributions having the property $\mathcal{P}$. I.e.\ $d_{TV}(D,D') > \epsilon$ for every $D' \in \mathcal{P}$. 


We define the probability of self-collision of the distribution $D$ by $||D||_2$. For a set $S$ of samples drawn from $D$, $\coll(S)$ defines the pairwise collision count between them. Consider $S_1, S_2\subset S$, the \emph{bipartite collision} of $D$ with respect to $S$ is defined by $\coll(S_1,S_2)$ is the number of collision between $S_1$ and $S_2$.

We will be using the count of collisions among sample points to test closeness to uniformity. The following lemma connects the collision probability and the distance to uniformity.
\begin{lemma}[\cite{BKR04}]
    Let $D$ be a distribution over $[n]$. If $\max_xD(x)\leq (1+\epsilon).\min_xD(x)$ then $||D||_2^2\leq (1+\epsilon^2)/n$. If $||D||_2^2\leq (1+\epsilon^2)/n$ then $d_{TV}(D,\mathcal{U})\leq \epsilon$.
    \label{observation}
\end{lemma}

One way to test the properties of distributions is to first learn an explicit description of the distribution. We now define the notion of flattened and reduced distributions that will be useful towards this end.

\begin{definition}[Flattened and reduced distributions] Let $D$ be a distribution over $[n]$, and there exists a set of partitions of the domain into $\ell$ disjoint intervals, $\mathcal{I}=\{I_j\}_{j=1}^{\ell}$. The flattened distribution $(D^f)^{\mathcal{I}}$ corresponding to $D$ and $\mathcal{I}$ is a distribution over $[n]$ defined as follows : for $j\in [\ell]$ and $i\in I_j$; $(D^f)^{\mathcal{I}}(i)=\frac{\sum_{t\in I_j}D(t)}{|I_j|}$. A reduced distribution $D^r$ is defined over $[\ell]$ such that $\forall i\in \ell, D^r(i)=D(I_i)$.
\end{definition}
If a distribution $D$ is $\epsilon$ close to its flattened distribution according to some partition $\{\mathcal{I}_j\}_{j=1}^{\ell}$, we refer $D$ to be $(\epsilon,\ell)$-flattened. We note that if a distribution is monotonically non-increasing, then its flattened distribution is also monotonically non-increasing but its reduced distribution is not necessarily the same.

The following folklore result shows that the empirical distribution is close to the actual distribution provided sufficient number of samples are taken.

\begin{lemma}[Folklore]
Given a distribution $D$ supported over $[n]$ and an interval partition $\mathcal{I}=\{ I_1,...,I_{\ell}\}$, using $S=O(\frac{\ell^2}{\epsilon^2}\log{\ell})$ points from SAMP, we can obtain an empirical distribution $\tilde{D}$ in the following way: $\forall I_j\in \mathcal{I};\Tilde{D}(I_j)=\frac{occ(S,I_j)}{|S|}$ ($occ(S,I_j)$ is the number of samples from $S$ lies inside $I_j$) over $[\ell]$ such that for all interval $I_j$, with probability at least $2/3$, $|D(I_j)-\Tilde{D}(I_j)|\leq \frac{\epsilon}{\ell}$. Moreover, let the flattened distribution of $D$ be $(D^f)^{\mathcal{I}}$ and the flattened distribution of $\Tilde{D}$ be $(\Tilde{D}^f)^{\mathcal{I}}$, we can say that $d_{TV}((D^f)^{\mathcal{I}},(\Tilde{D}^f)^{\mathcal{I}})<\epsilon$.
\label{lem:folklore-learn}
\end{lemma}
While designing a tester for monotonicity, we use the following theorem due to Birge (\cite{Birge})
\begin{lemma}[Oblivious partitioning \cite{Birge}]Let $D$ be a non-increasing distribution over $[n]$ and $\mathcal{I}=\{I_1,...,I_{\ell}\}$ is an interval partitioning of $D$ such that $|I_j|=(1+\epsilon)^j$, for $0<\epsilon<1 $, then $\mathcal{I}$ has the following properties,
\begin{itemize}
    \item $\ell=O(\frac{1}{\epsilon}\log{n\epsilon})$
    \item The flattened distribution corresponding to $\mathcal{I}$, $(D^f)^{\mathcal{I}}$ is $\epsilon$-close to $D$, or $D$ is $(\epsilon, \ell)$-flattened.
\end{itemize}
\end{lemma}


Next, we describe a data structure called the \countmin sketch which is used to estimate the frequencies of elements in a one-pass stream. It was introduced by Cormode et al (\cite{CM04}). As we are dealing with a one-pass streaming algorithm with a memory constraint, it would be important to store samples in less space. \countmin sketch uses hash functions to store frequencies of the stream elements in sublinear space and returns an estimate of the same. 
\begin{definition}[\countmin sketch] A \countmin (CM) sketch with parameters $(\epsilon,\delta)$ is represented by a
two-dimensional array counts with width $w$ and depth $d$: $count[1,1],..., count[d,w]$. Set $w=\frac{e}{\epsilon}$ and $d=\log{1/\delta}$. Each entry of the array is initially zero. Additionally, $d$ hash functions $h_1,...,h_d : \{1,...,n\}\rightarrow \{1,...,w\}$ chosen uniformly at random from a pairwise-independent family. The space requirement for the count min sketch is $wd$ words. The sketch can be queried for the frequency of an element from the universe $\mathcal{U}$ of elements, and will return an estimate of its frequency.
\end{definition}
The lemma below captures the fact that the frequency of any element $x_i$ can be estimated from a \countmin sketch.
\begin{lemma}[\cite{CM04}]
    Let $\{x_1,...,x_S\}$ be a stream of length $S$ and $f_{x_i}$ be the actual frequency of an element $x_i$. Suppose $\tilde{f}_{x_i}$ be the stored frequency in count min sketch, then the following is true with probability at least $(1-\delta)$,
        $f_{x_i}\leq \tilde{f}_{x_i}\leq f_{x_i}+\epsilon S$.
    \label{cmin}
\end{lemma}

\section{Identity testing in the streaming model}
In this section, we discuss identity testing problems in the streaming settings. First, we start with testing monotone distribution where the stream consists of samples drawn from \samp. Later we describe a identity tester that uses \pcond queries one at a time online fashion.
\subsection{Testing monotone distributions}
We start with testing monotone distribution for identity and closeness. Both algorithms work in the following way: Partition the domain of the monotone distributions using the oblivious decomposition, get the reduced distribution over $\ell=O(\log{n})/\epsilon$, test identity and closeness for the reduced distributions. The lemma below captures the fact of why it is sufficient to test reduced distributions in the case of monotone distributions.
\begin{lemma} [\cite{DDSVV13}] \label{reduced}
     Let $\mathcal{I}=\{I_1,...,I_{\ell}\}$ be a partition of $[n]$. Suppose $D_1$ and $D_2$ are two distributions over $[n]$ such that both of them are $(\epsilon,\ell)$ flattened according to $\mathcal{I}$ and $(D_1^r)^{\mathcal{I}}$, $(D_2^r)^{\mathcal{I}}$ are the reduced distributions respectively, then the following happens, $|d_{TV}(D_1, D_2)-d_{TV}((D_1^r)^{\mathcal{I}},(D_2^r)^{\mathcal{I}})|\leq 2\epsilon$. Moreover, if $D_1=D_2$, $(D_1^r)^{\mathcal{I}}=(D_2^r)^{\mathcal{I}}$.
\end{lemma}
As monotone distributions are $(\epsilon,\ell)$-flattened according to the oblivious decomposition, we can reduce the problem of testing a monotone distribution into testing the corresponding reduced distribution. For testing identity and closeness with a monotone distribution, it is feasible to run an identity tester and closeness tester over the reduced distributions. We use the following theorem for using one pass streaming identity tester.
\begin{theorem}[ Streaming identity tester \cite{DGKR19}] \label{stream identity}
   Let $D^*$ be an explicit distribution over $[\ell]$ and $D$ be an unknown distribution over $[\ell]$. There exists a single pass streaming identity tester that takes $O(\ell \log{\ell}/m\epsilon^4)$ samples from \samp oracle for $\log{\ell}/\epsilon^2\leq m\leq \ell^{9/10}$ and does the following: if $D=D^*$, with probability at least $2/3$, the algorithm outputs \accept, if $d_{TV}(D,D^*)>\epsilon$, it outputs \reject with probability at least $2/3$.
\end{theorem}
Similarly, for closeness testing, we have the following tester in one-pass streaming settings,
\begin{theorem} [Streaming closeness tester \cite{DGKR19}]\label{stream closeness}
    Let $D_1$ and $D_2$ be two unknown distributions over $[\ell]$. There exists a single pass streaming closeness tester that takes $O(\ell \sqrt{\log{\ell}}/\sqrt{m}\epsilon^2)$ samples from \samp oracle for $\log{\ell}\leq m\leq \tilde{\Theta}(min(\ell,\ell^{2/3}/\epsilon^{4/3}))$ and does the following: if $D_1=D_2$, with probability at least $2/3$, the algorithm outputs \accept if $d_{TV}(D_1, D_2)>\epsilon$, it outputs \reject with probability at least $2/3$.
\end{theorem}
For the sake of being self-contained, we describe the identity tester for monotone below,

\begin{algorithm}
     \caption{Testing identity monotone Streaming}
\SetKwInOut{Input}{Input}
        \SetKwInOut{Output}{Output}
\Input{Sample access to an unknown monotone distribution $D$, explicitly given monotone distribution $D^*$, $0<\epsilon\leq 1$, memory requirement $\log{\log{(n\epsilon)}}/\epsilon^2\leq m\leq (\log{n\epsilon}/\epsilon)^{9/10}$}
\Output{\accept if $D=D^*$, \reject if $d_{TV}(D,D^*)\geq 3\epsilon$}
Let $\mathcal{I}=\{I_1,...,I_{\ell}\}$ be the oblivious partitions of $D$, $D^*$ and $(D^r)^{\mathcal{I}}$, $(D^{*r})^{\mathcal{I}}$ be the reduced distributions over $[\ell]$\\
Sample $O(\log{n\epsilon}\log{\log{(n\epsilon)}}/m\epsilon^5)$ points from $(D^r)^{\mathcal{I}}$\\
Run Streaming identity tester according to Theorem \ref{stream identity} for $(D^r)^{\mathcal{I}}$ and $(D^{*r})^{\mathcal{I}}$ \label{str}\\
\eIf{Streaming identity tester accepts}
{\accept}
{\reject}
\end{algorithm}
We now build the correctness of the algorithm. The algorithm needs sample access from $(D^r)^{\mathcal{I}}$ which can be obtained by sampling a point $x$ according to $D$ and then returning $I_j$, where $x\in I_j$. By replacing $\ell=O(\frac{\log{n\epsilon}}{\epsilon})$ in Theorem \ref{stream identity}, we get the sample complexity to be $O(\log{n\epsilon}\log{(\log{(n\epsilon)}}/m\epsilon^5)$. If $D=D^*$, the streaming identity tester outputs \accept. If $d_{TV}(D,D^*)\geq 3\epsilon$, by Lemma \ref{reduced}, $d_{TV}((D^r)^{\mathcal{I}},(D^{*r})^{\mathcal{I}})\geq \epsilon$ and the streaming identity tester outputs \reject. Hence, the algorithm is indeed correct.

The structure of the closeness testing algorithm in streaming settings for monotone distributions is similar to the identity tester except for the fact that in Line \ref{str}, instead of streaming identity tester we have to run streaming closeness tester according to Theorem \ref{stream closeness}. The samples from two unknown reduced distributions $(D_1^r)^{\mathcal{I}}$ and $(D_1^r)^{\mathcal{I}}$ can be obtained by drawing samples from $D_1$ and $D_2$ respectively and transforming them into the samples of the reduced distributions. By replacing $\ell=O(\frac{\log{(n\epsilon)}}{\epsilon})$ in Theorem \ref{stream closeness}, we get the sample complexity to be $O(\log{(n\epsilon)} \sqrt{\log{\log{(n\epsilon)}}}/\sqrt{m}\epsilon^3)$, where $\log{\log{(n\epsilon)}}\leq m\leq \tilde{\Theta}(min(\log{(n\epsilon)}/\epsilon,\log^{2/3}{(n\epsilon)}/\epsilon^{2}))$.

\subsection{Testing identity in the streaming model using \pcond}

In this section, we revisit the identity testing problem using \pcond queries: given sample access and \pcond query access to an unknown distribution $D$ we have to test whether $D$ is identical to a fully specified distribution $D^*$ or they are $\epsilon$ far from each other. Canonne et al (\cite{CRS15}) address the problem and propose a \pcond query-based identity tester. In their algorithm, the domain of $D^*$ is divided into a set of "buckets" where the points are having almost the same weights. The algorithm samples $\tilde{O}(\log^2{n}/poly(\epsilon))$ points from $D$ and estimates the weight of each bucket. They prove if $D$ and $D^*$ are far then there exists at least one bucket where the weight of $D^*$ and weight of $\tilde{D}$ will differ. If not, then the algorithm runs a process called \compare to estimate the ratio of the weight of each pair of points $(y,z)$ where $y$ is taken from a set of samples drawn from $D^*$ and $z$ is taken from a set of samples according to $D$. The following lemma is used to compare the weights of two points.
\begin{lemma}[\cite{CRS15}]
Given as input two disjoint subsets of points $X, Y$ together with parameters $\eta \in (0,1], K\geq 1 $ and $\delta\in (0,\frac{1}{2}] $ as well as \cond query access to a distribution $D$, there exists a procedure \compare which estimates the ratio of the weights of two sets and either outputs a value $\rho >0$ or outputs High or Low and satisfies the following:
\begin{itemize} 
\item  If $D(X)/K\leq D(Y)\leq K\cdot D(X) $ then with probability at least $1-\delta$ the procedure outputs a value $\rho \in [1-\eta,1+\eta]D(Y)/D(X)$;
\item If $D(Y)>K\cdot D(X) $ then with probability at least $1-\delta$ the procedure outputs either High or a value $\rho \in [1-\eta,1+\eta]D(Y)/D(X)$;
\item If $D(Y)<D(X)/K $ then with probability at least $1-\delta$ the procedure outputs either Low or a value $\rho \in [1-\eta,1+\eta]D(Y)/D(X)$.
\end{itemize} The procedure performs $O(\frac{K\log{1/\delta}}{\eta^2} )$ conditional queries on the set $X\cup Y$.
\label{compare}
\end{lemma}
However, for storing $\tilde{O}(\log^2{n}/poly(\epsilon))$ samples for estimating the weights of the buckets, an $\tilde{O}(\log^3{n}/poly(\epsilon))$ space is required considering each sampled point takes $\log{n}$ bits of memory. As we are dealing with a memory constraint of $m$ bits, for $m<O(\log^3{n})$, implementing the algorithm is not memory efficient. We use the main idea of Canonne et al (\cite{CRS15}), but instead of storing all samples, we use the \countmin sketch data structure for storing the frequencies of the elements of the stream. Later, the frequencies are used to estimate the weight of each bucket. By choosing the parameters of the \countmin sketch suitably, the total space required for our algorithm is at most $O(m/\epsilon)$ bits. The main concept of our algorithm lies in the theorem below,
\begin{theorem}[Testing Identity \cite{CRS15}]
    \label{thm:condidentity}
  There exists an identity tester that uses an $\Tilde{O}(\log^4{n}/\epsilon^4)$ \pcond queries and does the following: for every pair of distributions $D, D^*$ over $[n]$, where $D^*$ is fully specified, the algorithm outputs \accept with probability at least $2/3$ if $D=D^*$ and outputs \reject with probability at least $2/3$ if $d_{TV}(D, D^*)\geq \epsilon$.
\end{theorem}
Before moving into the algorithm, we define the \emph{bucketization} technique according to (\cite{CRS15}). For an explicit distribution $D^*$, the domain is divided into $\ell$ buckets $\mathcal{B}=\{B_1,...,B_{\ell}\}$, where each bucket contains a set of points which satisfies $B_j=\{i\in [n]:2^{j-1}\eta/n\leq D^*(i)\leq 2^j\eta/n \}$ and $B_0=\{i\in [n]: D^*(i)<\eta/n\}$, where $\eta=\epsilon/c$ for $c$ to be a constant. The number of buckets $\ell=O(\lceil \log{n/\eta}+1\rceil +1)$.

We are now ready to present our \pcond query-based one-pass streaming algorithm for identity testing. Our algorithm and the correctness borrow from (\cite{CRS15}) with the extra use of \countmin sketches to improve the trade-off between the sample complexity and the space used.

\begin{theorem}
    The algorithm {\sc pcond identity testing streaming} uses an \\ $O(\log^2{n}\log{\log{n}}/m\epsilon^2)$ length stream of standard access query points and an $\tilde{O}(\log^4{n}/\epsilon^4)$ length of conditional stream and does the following,
    If $D=D^*$, it returns \accept with probability at least $2/3$, and if $d_{TV}(D, D^*)\geq \epsilon$, it returns \reject with probability at least $2/3$. The memory requirement for the algorithm is $O(\frac{m}{\epsilon})$ where $\frac{\log{n}\sqrt{\log{\log{n}}}}{{\epsilon}}\leq m\leq \frac{\log^2{n}}{\epsilon}$.
\end{theorem}
\begin{proof}
\textbf{Completeness}: Suppose $D=D^*$. We prove that the algorithm does not return \reject in Line \ref{lline}. Let $\tilde{D}(B_j)$ be the estimated weight of a bucket $B_j$ where  $\tilde{D}(B_j)=\frac{f_{B_j}}{S}$ for $S=O(\log^2{n}\log{\log{n}}/m\epsilon^2)$. An additive Chernoff bound [followed by a union bound over the buckets] shows that with high probability, $\forall B_j,|D(B_j)-\tilde{D}(B_j)|\leq \frac{\sqrt{m}\epsilon}{\log{n}}$.
    Using Lemma \ref{cmin}, with probability at least $99/100$, for every element $x_i$ in the stream, $f_{x_i}\leq \tilde{f}_{x_i}\leq  f_{x_i}+\frac{\epsilon S}{m}$. Summing over all the elements in a bucket $B_j$, we get $\tilde{f}_{B_j}-\frac{\epsilon}{m}S^2\leq f_{B_j}\leq  \tilde{f}_{B_j} $. Substituting $\tilde{D}(B_j)=\frac{f_{B_j}}{S}$, we can see that
         $\frac{\tilde{f}_{B_j}}{S}-\frac{\epsilon S}{m}\leq \tilde{D}(B_j)\leq  \frac{\tilde{f}_{B_j}}{S}$.
    As $D=D^*$, $\tilde{D}(B_j)$ is a good estimate of $D^*(B_j)$. Using $|D^*(B_j)-\tilde{D}(B_j)|\leq \frac{\sqrt{m}\epsilon}{\log{n}}$, we get 
     $\frac{\tilde{f}_{B_j}}{S}-\frac{\epsilon S}{m}-\frac{\sqrt{m}\epsilon}{\log{n}}\leq D^*(B_j)\leq \frac{\tilde{f}_{B_j}}{S} + \frac{\sqrt{m}\epsilon}{\log{n}}$.
This can be written as $ D^*(B_j)-\frac{\sqrt{m}\epsilon}{\log{n}} \leq \frac{\tilde{f}_{B_j}}{S}\leq D^*(B_j)+\frac{\sqrt{m}\epsilon}{\log{n}}+\frac{\log^2{n}\log{\log{n}}}{\epsilon m^2} $ by replacing $S=O(\log^2{n}\log{\log{n}}/m\epsilon^2)$. Hence, the algorithm will not output \reject with high probability. As $D=D^*$, for all pairs $(y_k,z_l)$ such that $\frac{D^*(y_k)}{D^*(z_l)}\in [1/2,2] $, it follows from Lemma \ref{compare} that the estimated ratio of weights of each pair $(y_k,z_l)$ is less than $(1-\eta/2\ell)\frac{D^*(y_k)}{D^*(z_l)}$ [for $\eta=\epsilon/6$] with probability at most $1/10s^2$. A union bound over all $O(s^2)$ pairs proves that with a probability of at least $9/10$ the algorithm outputs \accept.

\begin{algorithm}[H]
     \caption{\pcond \ Identity Testing Streaming}\label{alg: pcondidentity}
\SetKwInOut{Input}{Input}
        \SetKwInOut{Output}{Output}
\Input{\samp and \pcond access to $D$, an explicit distribution $D^*$, parameters $0<\epsilon\leq 1$, $\eta=\epsilon/6$, $\ell$ buckets of $D^*$, space requirement $O(m)$ bits $\frac{\log{n}\sqrt{\log{\log{n}}}}{{\epsilon}}\leq m\leq \frac{\log^2{n}}{\epsilon}$}
\Output{\accept if $D=D^*$, \reject if $d_{TV}(D,D^*)\geq \epsilon$} 
Sample $S=O(\frac{\log^2{n}\log{\log{n}}}{m\epsilon^2})$ points $\{x_1,...,x_S\}$ from \samp \\
\For{$(i=1$ to $S)$}
{Estimate the frequency of $x_i$ using \countmin sketch $(\frac{\epsilon}{m},\frac{1}{100})$ such that $f_{x_i}\leq \tilde{f}_{x_i}\leq f_{x_i}+\frac{\epsilon}{m}S$}
Define the frequency of each bucket $B_j$ to be $f_{B_j}=\sum_{x_i\in B_j}f_{x_i}$, such that $f_{B_j}\leq \tilde{f}_{B_j}\leq f_{B_j}+\frac{\epsilon}{m}S^2$\\
\If{$\frac{\tilde{f}_{B_j}}{S}<D^*(B_j)-\frac{\sqrt{m}\epsilon}{\log{n}}$ or $\frac{\tilde{f}_{B_j}}{S}>D^*(B_j)+\frac{\sqrt{m}\epsilon}{\log{n}} +\frac{\log^2{n}\log{\log{n}}}{\epsilon m^2}$\label{line5}} 
{\reject and Exit \label{lline}}
Select $s=O(\ell/\epsilon)$ points $\{y_1,...,y_s\}$ from $D^*$\label{line7}\\
\For{ each $y_k\in s$}
{Sample $s$ points $\{z_1,...,z_s\}$ from $D$ as a stream\label{line9}\\
\For{each pair of points $(y_k,z_l)$ such that $\frac{D^*(y_k)}{D^*(z_l)}\in [1/2,2] $}
{Run \compare($y_k,z_l,\eta/4\ell,2,1/10s^2$))\\
\If{Compare returns Low or a value smaller than $(1-\eta/2\ell)\frac{D^*(y_k)}{D^*(z_l)}$ \label{line12}}
{\reject and Exit}
}
}
\accept
\end{algorithm}

\textbf{Soundness :} Let $d_{TV}(D,D^*)\geq \epsilon$. In this case, if one of the estimates of $\tilde{f}_{B_j}$ passes Line \ref{line5}, the algorithm outputs \reject. Let's assume that the estimates are correct with high probability. The rest of the analysis follows from (\cite{CRS15}), we give a brief outline of the proof for making it self-contained. Define high-weight and low-weight buckets in the following way, for $\eta =\epsilon/6$, as follows: 
$    H_j =\{x\in B_j: D(x)>D^*(x)+\eta/\ell |B_j|\}$, and $     L_j =\{x\in B_j: D(x)\leq D^*(x)-\eta/\ell |B_j|\}$.
It can be shown that at least one point will occur from the low-weight bucket while sampling $s$ points in Line \ref{line7} and at least one point will come from the high-weight bucket while obtaining $s$ points in Line \ref{line9}. Using the definition of high-weight and low-weight buckets, there exists a pair $(y_k,z_l)$ such that $D(y_k)\leq (1-\eta/2\ell)D^*(y_k)$ and $D(z_k)> (1+\eta/2\ell)D^*(z_k)$. By Lemma \ref{compare}, with probability at least $1-1/10s^2$, \compare will return low or a value at most $(1-\eta/2\ell)\frac{D^*(y_k)}{D^*(z_l)}$ in Line \ref{line12}. Hence the algorithm outputs \reject with high probability.
\end{proof}

We use \countmin sketch with parameters $(\frac{\epsilon}{m},\frac{1}{100})$ in our algorithm. Comparing it with $(\epsilon,\delta)$ \countmin sketch defined in (\cite{CM04}), we set the width of the array to be $w=em/\epsilon$ and depth $d=\log{100}$. So the space required for the algorithm is $w.d$ words which imply $O(\frac{m}{\epsilon})$ bits. For running the \compare procedure, we are not using any extra space for storing samples. This is because for every element in $\{y_1,...,y_s\}$ we are sampling $s$ length stream $\{z_1,...,z_s\}$ and running \compare for each pair of points taken from each stream respectively. This leads to running compare process $s^2$ times. A single run of compare works in the following way in the streaming settings, for a pair $(y_k,z_l)$, sample $O(\log^2{n}/\epsilon^2)$ points from $D$ conditioned on $(y_k,z_l)$ and keep two counters for checking the number of times each of them appeared in the stream. Each round of \compare process requires $O(\log^2{n}/\epsilon^2)$ length of the stream. Hence, the total stream length is $\tilde{O}(\log^4{n}/\epsilon^4)$. 
\section{Testing monotonicity}
Monotonicity testing is a fundamental problem of distribution testing. 
 Given sample access to an unknown distribution $D$ over $[n]$, the task is to check whether $D$ is a monotone (non-increasing) or $\epsilon$ far from monotonicity. The problem was addressed by Batu et al (\cite{BKR04}) where they use the samples according to the standard access model. The algorithm divides the domain into half recursively until the number of collisions over an interval is less. A set of partitions is obtained this way and they are used to construct an empirical distribution by sampling another set of samples. Finally, the algorithm returns accept if the empirical distribution is close to monotone. The approach leads to an $O(\sqrt{n}\log{n}/\epsilon^4)$ SAMP query algorithm for this problem.

\subsection{Testing monotonicity using oblivious decomposition}\label{sec:mon-oblivious}
In this section, we discuss monotonicity testing using oblivious decomposition. Unlike Batu et al, we use the oblivious partitions for the unknown distribution $D$ and count the number of collisions over the intervals where enough sample lies. Our algorithm uses standard access queries to $D$ and proceeds by examining whether the total weight of such intervals is high or low. The insight of the algorithm is to apply oblivious decomposition instead of constructing the partitions recursively. An upper bound of $O(\sqrt {n}\log{n}/\epsilon^4)$ is obtained that matches the query complexity of Batu et al. The high-level idea of the algorithm is as follows: let $D$ be a monotone distribution and $\mathcal{I}=\{I_1,..., I_{\ell}\}$ be the oblivious decomposition for $\ell=O(\log{n}/\epsilon_1)$; where $\epsilon_1=O(\epsilon^2)$. If $(D^f)^{\mathcal{I}}$ is the flattened distribution, by oblivious decomposition, $d_{TV}(D,(D^f)^{\mathcal{I}})\leq \epsilon^2$ which simplifies to $  \sum_{j=1}^{\ell} D(I_j) d_{TV}(D_{I_j},\mathcal{U}_{I_j})\leq \epsilon^2$. We divide the partitions into two types, one, where the conditional distribution over $I_j$ is close to uniformity, and another, where the same is far from uniformity. In the case of monotone distribution, if the conditional distributions for a set of intervals are far from uniformity, the total weight over those intervals can not be too high. We also observe that the collision counts for such intervals are pretty high when enough sample lies in it. So, the problem of testing monotonicity boils down to the task of finding such high collision intervals and estimating their total weight. For the weight estimation, we obviously require an empirical distribution $\tilde{D}$ that can be constructed by sampling $poly(\log{n},1/\epsilon)$ samples. Later, the tester tests whether $\tilde{D}$ is close to monotone to reveal the final decision. When there are enough samples lying inside an interval, the pairwise collision count between them can be used to estimate the collision probability. The following lemma establishes the fact formally,
\begin{lemma} [\cite{BKR04}] \label{basic}
 Let $I\subset [n]$ be an interval of a distribution $D$ over $[n]$, $D_{I}$ be the conditional distribution of $D$ on $I$ and $S_{I}$ be the set of samples lying in interval $I$. Then,
  $$ ||D_{I}||_2^2-\frac{\epsilon^2}{64|I|} \leq \frac{coll(S_{I})}{{|S_{I}| \choose 2}}\leq ||D_{I}||_2^2+\frac{\epsilon^2}{64|I|}$$
with probability at least $1-O(\log^{-3}{n})$ provided that $|S_{I}|\geq O(\sqrt{|I|}\log{\log{n}}/\epsilon^4)$.
\end{lemma}

An interval partition of $[n]$ produces two types of intervals, where the conditional distributions are close and far from uniformity respectively. A pairwise collision count between samples is used to detect such intervals. Given enough samples lie in an interval, if the conditional distribution is far from uniformity, the collision count will be high. Similarly, the collision probability is low for the intervals where the conditional distributions are close to uniformity.

\begin{lemma}\label{lemmaa_1}
Let $D$ be an unknown distribution over $[n]$ and $I\subset [n]$ be an interval. Let $S_I$ be the set of samples lying inside $I$ such that $|S_{I}|\geq O(\sqrt{|I|}\log{\log{n}}/\epsilon^4)$, then the following happens
\begin{itemize}
\item If $d_{TV}(D_{I},\mathcal{U}_{I})>\frac{\epsilon}{4}$, then $\frac{coll(S_{I})} {{|S_{I}|\choose 2}}>\frac{1}{|I|}+\frac{\epsilon^2}{64|I|}$
\item If $d_{TV}(D_{I},\mathcal{U}_{I})\leq\frac{\epsilon}{4}$, then, $\frac{coll(S_{I})}{{|S_{I}| \choose 2}}\leq\frac{1+\epsilon^2/64}{|I|}+\frac{\epsilon^2}{16}$
\end{itemize}
\end{lemma}
\begin{proof}
Let, $d_{TV}(D_{I},\mathcal{U}_{I})  >\frac{\epsilon}{4}$, squaring both sides, we get $(d_{TV}(D_{I},\mathcal{U}_{I}))^2>\frac{\epsilon^2}{16}>\frac{\epsilon^2}{32}$. Using the fact that $d_{TV}(D_{I},\mathcal{U}_{I}) \leq \sqrt{|I|}\cdot||D_{I}-\mathcal{U}_{I}||_2$, we deduce $|I|\cdot||D_{I}-\mathcal{U}_{I}||_2^2>\frac{\epsilon^2}{32}$. Simplifying the inequality, we get $ ||D_{I}-\mathcal{U}_{I}||_2^2>\frac{\epsilon^2}{32|I|}$. Now, we obtain the following inequality by using  $||D_{I}-\mathcal{U}_{I}||_2^2= ||D_{I}||_2^2-\frac{1}{|I|} $.
    \begin{gather*}
    ||D_{I}||_2^2-\frac{1}{|I|}>\frac{\epsilon^2}{32|I|}\\
     ||D_{I}||_2^2>\frac{\epsilon^2}{32|I|}+\frac{1}{|I|}
    \end{gather*}

Considering $|S_{I}|\geq O(\sqrt{|I|}\log{\log{n}}/\epsilon^4)$, by Lemma \ref{basic}, $ ||D_{I}||_2^2>\frac{\epsilon^2}{32|I|}+\frac{1}{|I|}$ implies the following,
\begin{align*}
     \frac{coll(S_{I})}{{|S_{I}| \choose 2}}+\frac{\epsilon^2}{64|I|}&>\frac{\epsilon^2}{32|I|}+\frac{1}{|I|}\\
    \frac{coll(S_{I})}{{|S_{I}| \choose 2}}  &>\frac{1}{|I|}+\frac{\epsilon^2}{64|I|}
\end{align*}
Similarly, when $d_{TV}(D_{I},\mathcal{U}_{I})\leq\frac{\epsilon}{4}$, we get $||D_{I}||_2^2\leq\frac{\epsilon^2}{16}+\frac{1}{|I|}$. Given $|S_{I}|\geq O(\sqrt{|I|}\log{\log{n}}/\epsilon^4)$, by Lemma \ref{basic}, 
     $\frac{coll(S_{I})}{{|S_{I}| \choose 2}}  \leq \frac{1+\epsilon^2/64}{|I|}+\frac{\epsilon^2}{16}$.
\end{proof}

We are now ready to present the collision monotonicity tester.
\begin{algorithm}
\caption{Collision Monotonicity} \label{collision_monotone_tester}
\SetKwInOut{Input}{Input}
        \SetKwInOut{Output}{Output}
\Input{\samp access to $D$, $\ell=O(\frac{1}{\epsilon_1}\log{(n\epsilon_1+1)})$ oblivious partitions $\mathcal{I}=\{I_1,..,I_{\ell}\}$ and error parameter $\epsilon,\epsilon_1\in (0,1]$, where $\epsilon_1=\epsilon^2$}
Sample $T=O(\frac{1}{\epsilon^6}\log^2{n}\log{\log{n}})$ points from \samp\\
Get the empirical distribution $\Tilde{D}$ over $\ell$\\
Obtain an additional sample $S=O(\frac{\sqrt{n}\log{n}\log{\log{n}}}{\epsilon^8})$ from \samp \\
Let $J$ be the set of intervals where the number of samples (in each interval $I_j$) is $|S_{I_j}|\geq O(\sqrt{|I_j|}\log{\log{n}}/\epsilon^4)$ and $coll(S_{I_j})\geq(\frac{1+\epsilon^2/64}{|I_j|}+\frac{\epsilon^2}{16}) {|S_{I_j}| \choose 2} $\\
\If{$\sum_{I_j\in J} \Tilde{D}(I_j)>5\epsilon$}{\reject and Exit}\label{line_6}
Define a flat distribution $(\Tilde{D}^f)^{\mathcal{I}}$ over $[n]$ using $\Tilde{D}$\\
Output \accept if $(\Tilde{D}^f)^{\mathcal{I}}$ is $2\epsilon$ close to a monotone distribution. Otherwise output \reject \label{algo_line1}
\end{algorithm}

\begin{theorem}\label{collision_monotone}
The algorithm {\sc collision monotonicity} uses $O(\frac{\sqrt{n}\log{n}\log{\log{n}}}{\epsilon^8})$ \samp queries and outputs \accept with probability at least $2/3$ if $D$ is a monotone distribution and outputs \reject with probability at least $2/3$ when $D$ is not $7\epsilon$-close to monotone. 
\end{theorem}
\begin{proof}
We start by defining an interval to be high weight, if $D(I_j)\geq \epsilon^2/\log{n}$. As there are $O(\frac{\log{n}}{\epsilon^2})$ partitions exist, there will be at least one such interval with $D(I_j)\geq \frac{\epsilon^2}{\log{n}} $. Moreover, an additive Chernoff bound shows that all such high-weight intervals will contain $|S_{I_j}|\geq O(\sqrt{|I_j|}\log{\log{n}}/\epsilon^4)$ samples while sampling $O(\frac{\sqrt{n}\log{n}\log{\log{n}}}{\epsilon^8})$ points according to $D$.

\textbf{Completeness :} Let $D$ be monotone, then by oblivious partitioning, we have, \\ $\sum_{j=1}^{\ell} \sum_{x\in I_j}|D(x)-\frac{D(I_j)}{|I_j|}|\leq \epsilon_1$. By simplifying, we get $\sum_{j=1}^{\ell} D(I_j)\sum_{x\in I_j}|\frac{D(x)}{D(I_j)}-\frac{1}{|I_j|}|\leq \epsilon^2 $ which implies $\sum_{j=1}^{\ell} D(I_j) d_{TV}(D_{I_j},\mathcal{U}_{I_j})\leq \epsilon^2 $.
Let $J'$ be the set of intervals where for all $I_j$, $d_{TV}(D_{I_j},\mathcal{U}_{I_j})>\frac{\epsilon}{4}$, then we have,
$\sum_{I_j\in J'} D(I_j) \leq 4\epsilon$. Let $\hat{J}$ be the set of intervals where $|S_{I_j}|\geq O(\sqrt{|I_j|}\log{\log{n}}/\epsilon^4)$ and  $d_{TV}(D_{I_j},\mathcal{U}_{I_j})  >\frac{\epsilon}{4}$. So, $\hat{J}\subseteq J'$. From Lemma \ref{lemmaa_1}, we know $\hat{J}$ is the set of intervals where $\frac{coll(S_{I_j})} {{|S_{I_j}|\choose 2}}>\frac{1}{|I_j|}+\frac{\epsilon^2}{64|I_j|}$. Let $J$ be the set of intervals where $|S_{I_j}|\geq O(\sqrt{|I_j|}\log{\log{n}}/\epsilon^4)$ and $\frac{coll(S_{I_j})}{{|S_{I_j}| \choose 2}}\geq\frac{1+\epsilon^2/64}{|I_j|}+\frac{\epsilon^2}{16} $, then $J\subseteq \hat{J}\subseteq J'$. We have already proved $\sum_{I_j\in J'} D(I_j) \leq 4\epsilon$. So, we can conclude that $\sum_{I_j\in J} D(I_j) \leq 4\epsilon$.

When $d_{TV}(D_{I_j},\mathcal{U}_{I_j})\leq\frac{\epsilon}{4}$, the algorithm does not sum over such $D(I_j)$ even if the number of samples $|S_{I_j}|\geq O(\sqrt{|I_j|}\log{\log{n}}/\epsilon^4)$ as $ \frac{coll(S_{I_j})}{{|S_{I_j}| \choose 2}}\leq\frac{1+\epsilon^2/64}{|I_j|}+\frac{\epsilon^2}{16}$ by Lemma \ref{lemmaa_1}. As a result, we can say
 $\sum_{I_j\in J} D(I_j) \leq 4\epsilon$ when $D$ is monotone. The rest of the proof follows from the fact that we have a distribution $\Tilde{D}$ over $\ell$ such that by Lemma \ref{lem:folklore-learn} $\forall I_j$, $ |D(I_j)-\Tilde{D}(I_j)|\leq \frac{\epsilon}{\ell}$.
Substituting the value of $\ell=O(\frac{1}{\epsilon^2}\log{n})$, we get, $\forall I_j$
 $D(I_j)- \frac{\epsilon^3}{\log{n}}  \leq \Tilde{D}(I_j)\leq D(I_j)+ \frac{\epsilon^3}{\log{n}}$. Summing over all intervals $J$ such that $|S_{I_j}|\geq O(\sqrt{|I_j|}\log{\log{n}}/\epsilon^4)$ and $\frac{coll(S_{I_j})}{{|S_{I_j}| \choose 2}}\geq\frac{1+\epsilon^2/64}{|I_j|}+\frac{\epsilon^2}{16} $, we deduce, $\sum_{I_j\in J}\Tilde{D}(I_j)\leq \sum_{I_j\in J}D(I_j)+\sum_{I_j\in J} \frac{\epsilon^3}{\log{n}}$. We have proved that $\sum_{I_j\in J}D(I_j)\leq 4\epsilon$ when $D$ is monotone. Using this fact, we obtain $\sum_{I_j\in J}\Tilde{D}(I_j)\leq \sum_{I_j\in J}D(I_j)+ \epsilon$ which implies $\sum_{I_j\in J}\Tilde{D}(I_j)\leq 5\epsilon$. Hence, the algorithm will NOT output \reject in Step \ref{line_6}.

By Birge's decomposition with parameter $\epsilon_1=\epsilon^2$, we get $d_{TV}(D,(D^f)^{\mathcal{I}})\leq \epsilon^2$. By Lemma \ref{lem:folklore-learn}, we know $d_{TV}((D^f)^{\mathcal{I}},(\Tilde{D}^f)^{\mathcal{I}})<\epsilon$. Hence, by triangle inequality, we obtain $d_{TV}(D,(\Tilde{D}^f)^{\mathcal{I}})<\epsilon+\epsilon_1<2\epsilon$. This implies that the flattened distribution $(\Tilde{D}^f)^{\mathcal{I}}$ is $2\epsilon$ close to a monotone distribution as $D$ is monotone. Hence, the algorithm will output \accept in Step \ref{algo_line1}.

\textbf{Soundness :} We will prove the contrapositive of the statement. Let the algorithm outputs \accept, then we need to prove that $D$ is $7\epsilon$ close to monotone. As the algorithm outputs accept, $\sum_{I_j\in J}\Tilde{D}(I_j)\leq 5\epsilon$, where $J$ is the set of intervals with
$|S_{I_j}|\geq O(\sqrt{|I_j|}\log{\log{n}}/\epsilon^4)$ and $coll(S_{I_j})\geq(\frac{1+\epsilon^2/64}{|I_j|}+\frac{\epsilon^2}{16}) {|S_{I_j}| \choose 2} $. When $\frac{coll(S_{I_j})}{{|S_{I_j}| \choose 2}}\geq\frac{1+\epsilon^2/64}{|I_j|}+\frac{\epsilon^2}{16} $ by Lemma \ref{basic}, we get $||D_{I_j}||_2^2+\frac{\epsilon^2}{64|I_j|}\geq \frac{1+\epsilon^2/64}{|I_j|}+\frac{\epsilon^2}{16}$. Using $||D_{I}-\mathcal{U}_{I}||_2^2= ||D_{I}||_2^2-\frac{1}{|I|}$, we obtain $     \frac{1}{|I_j|}+||D_{I_j}-\mathcal{U}_{I_j}||_2^2\geq \frac{1}{|I_j|}+\frac{\epsilon^2}{16}$. Finally we deduce $d_{TV}(D_{I_j},\mathcal{U}_{I_j}) \geq \frac{\epsilon}{4} $ by using the fact that $d_{TV}(D_{I_j},\mathcal{U}_{I_j}) \geq ||D_{I_j}-\mathcal{U}_{I_j}||_2$.
 
Now, we calculate the distance between $D$ and $(D^f)^{\mathcal{I}}$ corresponding to $\mathcal{I}$,
\begin{align*}
    d_{TV}(D,(D^f)^{\mathcal{I}})&=\frac{1}{2}\sum_{j=1}^{\ell}\sum_{x\in I_j}|D(x)-\frac{D(I_j)}{|I_j|}|\\
    &=\frac{1}{2}\Big[\sum_{\frac{\epsilon}{4}\leq d_{TV}(D_{I_j},\mathcal{U}_{I_j})\leq 1 }D(I_j)d_{TV}(D_{I_j},\mathcal{U}_{I_j})+ \sum_{d_{TV}(D_{I_j},\mathcal{U}_{I_j}) < \frac{\epsilon}{4}} D(I_j)  d_{TV}(D_{I_j},\mathcal{U}_{I_j})\Big]\\
    &< \frac{1}{2}\Big[\sum_{\frac{\epsilon}{4}\leq d_{TV}(D_{I_j},\mathcal{U}_{I_j})\leq 1 }\Big(\Tilde{D}(I_j)+\frac{\epsilon^3}{\log{n}}\Big)+ \frac{\epsilon}{4}\Big]\\
    &<\frac{1}{2} \Big[6\epsilon +\frac{\epsilon}{4}\Big]\\
    &< 4\epsilon
\end{align*}
The second inequality is using the fact that $\forall I_j, \Tilde{D}(I_j)\leq D(I_j)+\frac{\epsilon^3}{\log{n}}$. As the algorithm outputs accept we have, $\sum_{I_j\in J}\Tilde{D}(I_j)\leq 5\epsilon$. Summing over at most $\ell=O(\log{n}/\epsilon^2)$ intervals, we get the above result.

For the rest of the proof, we use the fact that $d_{TV}((D^f)^{\mathcal{I}},(\Tilde{D}^f)^{\mathcal{I}})<\epsilon $ by Lemma \ref{lem:folklore-learn}.
Using triangle inequality, we get $d_{TV}(D,(\Tilde{D}^f)^{\mathcal{I}})<5\epsilon$. As the algorithm outputs \accept, there exists a monotone distribution $M$, such that $d_{TV}((\Tilde{D}^f)^{\mathcal{I}}, M)\leq 2\epsilon$. Again by applying triangle inequality, we have $d_{TV}(D,M)<7\epsilon$. This completes the proof.
\end{proof}


\subsection{Testing Monotonicity using Bipartite Collisions}
In this section, we perform the monotonicity testing in a slightly different fashion which functions as the building block of a streaming-based monotonicity tester. Here, unlike counting pairwise collisions between the samples, we divide the samples into two sets and count the bipartite collisions between them. The idea of the bipartite collision tester is adapted from (\cite{DGKR19}). A key Lemma \ref{basic_0} proves how the bipartite collision is used to estimate the collision probability. 
Given sample access to an unknown distribution $D$ over $[n]$, first, we divide the domain according to the oblivious decomposition. We count the bipartite collisions inside the intervals where enough samples lie. If $D$ is monotone, the total weight of high collision intervals can not be too high. We estimate the total weight of such intervals and the rest of the algorithm works similarly to the collision monotonicity tester presented in Algorithm \ref{collision_monotone_tester}. Prior to describing the algorithm, the lemma below clarifies the fact "enough samples" and the intervals holding them. 

\begin{lemma}
Let $D$ be a distribution over $[n]$, and $\mathcal{I}=\{I_1,...,I_{\ell}\}$ be an interval partitions of $[n]$. Let $\mathcal{J}\subset \mathcal{I}$ be the set of intervals and for all $I_j\in \mathcal{J}$, $D(I_j)\geq \epsilon_1/\log{n}$, where $\epsilon_1=\epsilon^2$. If $S=O(\frac{n\log{n}}{\epsilon^8})$ samples are drawn according to $D$, then all $I_j\in \mathcal{J}$ contain $|S_{I_j}|\geq O(|I_j|/\epsilon^4)$ samples.
\label{lem:chernoff}
\end{lemma}
\begin{proof}
    Fix an $I_j$ and define a random variable, $X_i=1$ if $i^{th}$ sample is in $I_j$ else $0$. Let $X=\sum_{i=1}^S X_i=S_{I_j}$. Then the expectation $\mathbb{E}[X]=|S|\cdot D(I_j)\geq \frac{|S|\epsilon_1 }{\log{n}}$.

By Chernoff bound, we can see that
    $Pr\Big[  X<(1-\epsilon)  \frac{|S|\epsilon_1 }{\log{n}}       \Big]= Pr\Big[  X<(1-\epsilon)\mathbb{E}[X]\Big]\leq e^{-\epsilon^2\mathbb{E}[X]}
    \leq e^{-\epsilon^2\frac{|S|\epsilon^2}{\log{n}}}
    <\frac{\epsilon^2}{10\log{n}}$.

The last inequality is obtained from the fact that $|S|=O(\frac{n\log{n}}{\epsilon^8})$ and using $\frac{n}{\epsilon^4}>\log({10\log{n}/\epsilon^2)} $. Applying union bound over all $\ell=O(\frac{\log{n}}{\epsilon_1})$ partitions, we can conclude that, $[\epsilon_1=\epsilon^2]$ $\forall I_j$; such that $D(I_j)\geq \frac{\epsilon_1}{\log{n}}$ with probability at least $9/10$, the following happens,
$S_{I_j}\geq (1-\epsilon)\frac{|S|\epsilon_1 }{\log{n}}
\geq (1-\epsilon)\frac{n}{\epsilon^6}
\geq O(|I_j|/\epsilon^4)$  
\end{proof}

The main intuition behind our algorithm is counting the bipartite collision between a set of samples. The next lemma, 
defines the necessary conditions for estimating collision probability using bipartite collision count.
\begin{lemma}\label{basic_0}
Let $D$ be an unknown distribution over $[n]$ and $S$ be the set of samples drawn according to \samp. Let $I\subset [n]$ be an interval and $S_{I}$ be the set of points lying in the interval $I$. Let $S_I$ be divided into two disjoint sets $S_1$ and $S_2$; $\{S_1\}\cup \{S_2\}=\{S_I\}$ such that $|S_1||S_2|\geq O(|S_I|/\epsilon^4)$, then with probability at least $2/3$,

    $$ ||D_{I}||_2^2-\frac{\epsilon^2}{64|I|} \leq \frac{coll(S_1,S_2)}{|S_1||S_2|}\leq ||D_{I}||_2^2+\frac{\epsilon^2}{64|I|}.$$
\end{lemma}

\begin{proof}
Define the random variable $X_i=1$ if $i^{th}$ sample in $S_1$ is same as $j^{th}$ sample in $S_2$, $0$ otherwise.
\begin{align*}
X&=\sum_{(i,j)\in S_1\times S_2} X_{ij}=coll(S_1,S_2)\\
    \mathbb{E}[X]&=|S_1|\cdot|S_2|\cdot||D_{I}||_2^2
\end{align*}
Where $||D_{I}||_2^2 $ is collision probability. Let $Y_{ij}=X_{ij}-\mathbb{E}[X_{ij}]=X_{ij}-||D_{I}||^2_2$.
\begin{align*}
    Var[\sum_{(i,j)\in S_1\times S_2} X_{ij}]&=\mathbb{E}\Big[(\sum_{(i,j)\in S_1\times S_2} Y_{ij})^2\Big]\\
    &=\mathbb{E}\Big[\sum_{(i,j)\in S_1\times S_2} Y_{ij}^2+\sum_{(i,j)\neq (k,l);|\{i,j,k,l\}|=3} Y_{ij}Y_{kl}\Big]
\end{align*}
We calculate the following,
\begin{align*}
    \mathbb{E}[ Y_{ij}^2]&=\mathbb{E}[X_{ij}^2]-2(\mathbb{E}[X_{ij}])^2+(\mathbb{E}[X_{ij}])^2\\
    &=||D_{I}||_2^2-||D_{I}||_2^4\\
    \mathbb{E}[Y_{ij}Y_{kl}]&=\mathbb{E}\Big[ (X_{ij}-||D_{I}||^2_2)(X_{kl}-||D_{I}||^2_2)     \Big]\\
    &=\mathbb{E}\Big[ X_{ij}X_{kl}   \Big]-||D_{I}||_2^2(\mathbb{E}[X_{ij}]+\mathbb{E}[X_{kl}])+||D_{I}||_2^4\\
    &=\mathbb{E}\Big[ X_{ij}X_{kl}   \Big]-||D_{I}||_2^4
\end{align*}
Now, 
\begin{align*}
    Var[\sum_{(i,j)\in S_1\times S_2} X_{ij}]&=\sum_{(i,j)\in S_1\times S_2} (||D_{I}||_2^2-||D_{I}||_2^4)+\sum_{(i,j)\neq (k,l);|\{i,j,k,l\}|=3} (\mathbb{E}\Big[ X_{ij}X_{kl}   \Big]-||D_{I}||_2^4)\\
    &= |S_1|.|S_2|(||D_{I}||_2^2-||D_{I}||_2^4)  +\sum_{(i,j);(k,j)\in S_1\times S_2; i\neq k} \mathbb{E}\Big[X_{ij}X_{kj}\Big] \\
    &+\sum_{(i,j);(i,l)\in S_1\times S_2; j\neq l} \mathbb{E}\Big[X_{ij}X_{il}\Big]-\sum_{(i,j)\neq (k,l);|\{i,j,k,l\}|=3} ||D_I||^4_2\\
    &= |S_1|.|S_2|(||D_{I}||_2^2-||D_{I}||_2^4)+|S_2|{|S_1|\choose 2}||D_I||_3^3\\
    &+|S_1|{|S_2|\choose 2}||D_I||_3^3-\Big(|S_2|{|S_1|\choose 2}+|S_1|{|S_2|\choose 2}\Big)||D_{I}||_2^4\\
    &\leq |S_1||S_2|\Big[ (||D_I||_2^2-||D_{I}||_2^4)+(|S_1|+|S_2|)(||D_I||_3^3- ||D_{I}||_2^4) \Big]
\end{align*}
Applying Chebyshev's inequality, we get,
\begin{align*}
    Pr[|X-\mathbb{E}[X]|>\frac{\epsilon^2}{64|I|}|S_1||S_2|]&\leq \frac{64^2Var[X]|I|^2}{\epsilon^4|S_1|^2|S_2|^2}\\
    & \leq \frac{  |S_1||S_2|\Big[ (||D_I||_2^2-||D_{I}||_2^4)+(|S_1|+|S_2|)(||D_I||_3^3- ||D_{I}||_2^4) \Big]    64^2 |I|^2}{\epsilon^4|S_1|^2|S_2|^2}\\
    &\leq \frac{\Big[ ||D_I||_2^2-||D_I||_2^4+(|S_1|+|S_2|)(||D_I||_2^3-||D_I||_2^4) \Big]64^2 |I|^2}{\epsilon^4|S_1|\cdot|S_2|}\\
    &\leq \frac{\Big[ ||D_I||_2^2-||D_I||_2^4+(|S_1|+|S_2|)(||D_I||_2^2-||D_I||_2^4) \Big]64^2 |I|^2}{\epsilon^4|S_1|\cdot|S_2|}\\
    & \leq \frac{ ||D_I||_2^2\Big[1-||D_I||_2^2+(|S_1|+|S_2|)(1-||D_I||_2^2) \Big]64^2 |I|^2}{\epsilon^4|S_1|\cdot|S_2|}\\
    &\leq \frac{ ||D_I||_2^2\Big(1-||D_I||_2^2\Big)\Big(1+|S_1|+|S_2|\Big)64^2 |I|^2}{\epsilon^4|S_1|.|S_2|}
\end{align*}
Where the third inequality uses the fact that $||D_I||_3\leq ||D_I||_2$ and the fourth inequality uses the fact that  $||D_I||_2^3\leq ||D_I||_2^2$ as $||D_I||_2\in (0,1]$. To make the probability $<1/3$, we have,
\begin{align*}
    |S_1|\cdot|S_2|&\geq 3\times 64^2|I|^2\frac{1}{\epsilon^4}||D_I||_2^2\Big(1-||D_I||_2^2\Big)\Big(1+|S_1|+|S_2|\Big)\\
    &\geq 3\times 64^2\frac{|I|^2}{\epsilon^4}||D_I||_2^2\frac{||D_{I}||_2^2}{100}\Big(|S_1|+|S_2|\Big)\\
    &\geq 3\times 64^2\frac{1}{100\epsilon^4}\Big(|S_1|+|S_2|\Big)\\
    &\geq O(\frac{S_I}{\epsilon^4})
\end{align*}
In the second inequality we have used the fact that $(1-||D_I||_2^2)\geq \frac{1}{100}||D_I||_2^2 $ as $||D_I||_2^2\leq \frac{100}{101}<1$. The third inequality is obtained from the fact that $||D||_2^2\geq \frac{1}{|I|}$. The final inequality is obtained from the fact that $|S_I|=|S_1|+|S_2|$.
Therefore, provided $|S_1|\cdot |S_2|\geq O(\frac{|S_I|}{\epsilon^4})$, with probability at least $2/3$, $||D_{I}||_2^2-\frac{\epsilon^2}{64|I|} \leq \frac{coll(S_1,S_2)}{|S_1||S_2|}\leq ||D_{I}||_2^2+\frac{\epsilon^2}{64|I|}$.
\end{proof}

The bipartite collision-based tester works by verifying the total weight of the intervals where the conditional distributions are far from uniformity. Let $S_I$ be the set of samples inside an interval $I$ and let it satisfy the condition of Lemma \ref{basic_0}. The following lemma shows that bipartite collision count is used to detect such intervals.

\begin{lemma}\label{lemmaa_10}
Let $D$ be an unknown distribution over $[n]$ and $I\subset [n]$ is an interval. Let $S_{I}$ be the set of points lying in the interval $I$ and $S_{I}$ can be divided into two sets $S_1$ and $S_2$ such that $|S_1||S_2|\geq O(|S_I|/\epsilon^4)$, then the following happens with probability at least $2/3$  
\begin{itemize}
\item If $d_{TV}(D_{I},\mathcal{U}_{I})>\frac{\epsilon}{4}$, then $\frac{coll(S_1,S_2)}{|S_1||S_2|}>\frac{1}{|I|}+\frac{\epsilon^2}{64|I|}$
\item If $d_{TV}(D_{I},\mathcal{U}_{I})\leq\frac{\epsilon}{4}$, then, $\frac{coll(S_1,S_2)}{|S_1||S_2|}\leq\frac{1+\epsilon^2/64}{|I|}+\frac{\epsilon^2}{16}$
\end{itemize}
\end{lemma}

\begin{proof}
The proof is similar to the proof of Lemma \ref{lemmaa_1}. When $d_{TV}(D_{I},\mathcal{U}_{I})  >\frac{\epsilon}{4}$ we get $||D_{I}||_2^2>\frac{\epsilon^2}{32|I|}+\frac{1}{|I|}$. Consider $S_I$ is divided into two sets so that $|S_1|\cdot|S_2|\geq O(|S_I|/\epsilon^4)$, by Lemma \ref{basic_0} we obtain,
\begin{align*}
 \frac{coll(S_1,S_2)}{|S_1||S_2|} +\frac{\epsilon^2}{64|I|}&>\frac{\epsilon^2}{32|I|}+\frac{1}{|I|}\\
   \frac{coll(S_1,S_2)}{|S_1||S_2|} &>\frac{1}{|I|}+\frac{\epsilon^2}{64|I|}
\end{align*}
Similarly, when $d_{TV}(D_{I},\mathcal{U}_{I})\leq\frac{\epsilon}{4}$, we get $||D_{I}||_2^2\leq\frac{\epsilon^2}{16}+\frac{1}{|I|}$. Given $S_I$ can be divided into two sets such that $|S_1|\cdot|S_2|\geq O(|S_I|/\epsilon^4)$, by Lemma \ref{basic_0},
     $\frac{coll(S_1,S_2)}{|S_1|\cdot|S_2|}  \leq \frac{1+\epsilon^2/64}{|I|}+\frac{\epsilon^2}{16}$.
\end{proof}

Now, we present the bipartite collision-based monotonicity tester.
\begin{algorithm}
\caption{Bipartite Collision Monotonicity}\label{bipartite_tester}
\SetKwInOut{Input}{Input}
        \SetKwInOut{Output}{Output}
\Input{\samp access to $D$, $\ell=O(\frac{1}{\epsilon_1}\log{(n\epsilon_1+1)})$ oblivious partitions $\mathcal{I}=\{I_1,..,I_{\ell}\}$ and error parameter $\epsilon,\epsilon_1\in (0,1]$, where $\epsilon_1=\epsilon^2$}
\Output{\accept if $D$ is monotone, \reject if $D$ is not $7\epsilon$ close to monotone}
Sample $T=O(\frac{1}{\epsilon^6}\log^2{n}\log{\log{n}})$ points from \samp\\
Get the empirical distribution $\Tilde{D}$ over $\ell$\\
Obtain an additional sample $S=O(\frac{n\log{n}}{\epsilon^8})$ from \samp\\
Let $J$ be the set of intervals where the number of samples (in each interval $I_j$) is $|S_{I_j}|\geq O(|I_j|/\epsilon^4)$ and $S_{I_j}$ can be partitioned into two disjoint sets $S_1$ and $S_2$ such that $|S_1||S_2|\geq O(|I_j|/\epsilon^8)$ and
$\frac{coll(S_1,S_2)}{|S_1||S_2|} \geq(\frac{1+\epsilon^2/64}{|I_j|}+\frac{\epsilon^2}{16})  $\\
\If{$\sum_{I_j\in J} \Tilde{D}(I_j)>5\epsilon$}{\reject and Exit}\label{line_60}
Define a flat distribution $(\Tilde{D}^f)^{\mathcal{I}}$ over $[n]$\\
Output \accept if $(\Tilde{D}^f)^{\mathcal{I}}$ is $2\epsilon$-close to a monotone distribution. Otherwise output \reject \label{algo_line2}
\end{algorithm}
\begin{theorem} \label{bipartite}
The algorithm {\sc bipartite collision monotonicity} uses $O(\frac{n\log{n}}{\epsilon^8})$ \samp queries and outputs \accept with probability at least $2/3$ if $D$ is a monotone distribution and outputs \reject with probability at least $2/3$ when $D$ is not $7\epsilon$-close to monotone. 
\end{theorem}
\begin{proof}
While sampling $O(n\log{n}/\epsilon^8)$ points according to $D$, an application of Chernoff bound shows 
that the intervals with $D(I_j)\geq \epsilon^2/\log{n}$ will contain at least $S_{I_j}=O(|I_j|/\epsilon^4)$ points. There will be at least one such interval with $D(I_j)\geq \epsilon^2/\log{n}$ as there are $O(\log{n}/\epsilon^2)$ partitions.

\textbf{Completeness :} Let $D$ be monotone. By oblivious partitioning with parameter $\epsilon_1=\epsilon^2$, we have $\sum_{j=1}^{\ell} \sum_{x\in I_j}|D(x)-\frac{D(I_j)}{|I_j|}|\leq \epsilon_1$ which implies
        $\sum_{j=1}^{\ell} D(I_j) d_{TV}(D_{I_j},\mathcal{U}_{I_j})\leq \epsilon^2$. Let $J'$ be the set of intervals where for all $I_j$, $d_{TV}(D_{I_j},\mathcal{U}_{I_j})>\frac{\epsilon}{4}$, then $\sum_{I_j\in J'} D(I_j) \leq 4\epsilon$.

Let $\hat{J}$ is the set of intervals where $|S_1||S_2|\geq O(|S_{I_j}|/\epsilon^4)$ and  $d_{TV}(D_{I_j},\mathcal{U}_{I_j})  >\frac{\epsilon}{4}$. So, $\hat{J}\subseteq J'$. From Lemma \ref{lemmaa_10}, we know $\hat{J}$ is the set of intervals where $\frac{coll(S_1,S_2)}{|S_1||S_2|}>\frac{1}{|I_j|}+\frac{\epsilon^2}{64|I_j|}$. Let $J$ be the set of intervals where $|S_1||S_2|\geq O(|S_{I_j}|/\epsilon^4)$ and $\frac{coll(S_1,S_2)}{|S_1||S_2|}>\frac{1+\epsilon^2/64}{|I_j|}+\frac{\epsilon^2}{16} $, then $J\subseteq \hat{J}\subseteq J'$. We know $\sum_{I_j\in J'} D(I_j) \leq 4\epsilon$. So, we can conclude that $\sum_{I_j\in J} D(I_j) \leq 4\epsilon$.

When $d_{TV}(D_{I_j},\mathcal{U}_{I_j})\leq\frac{\epsilon}{4}$, the algorithm does not sum over such $D(I_j)$ even if $|S_1||S_2|\geq O(|S_{I_j}|/\epsilon^4)$. This is because by Lemma \ref{lemmaa_10} we know $\frac{coll(S_1,S_2)}{|S_1||S_2|}\leq\frac{1+\epsilon^2/64}{|I_j|}+\frac{\epsilon^2}{16}$. As a result, we can say that when $D$ is monotone $\sum_{I_j\in J} D(I_j) \leq 4\epsilon$. 

We use the empirical distribution $\tilde{D}$ and deduce that $\sum_{I_j\in J}\Tilde{D}(I_j)\leq 5\epsilon $. Hence, the algorithm will NOT output \reject in Step \ref{line_60}.
We also conclude as $D$ is monotone, the flattened distribution $(\Tilde{D}^f)^{\mathcal{I}}$ is $2\epsilon$ close to monotone and the algorithm will output \accept in Step \ref{algo_line2}.

\textbf{Soundness :} We will prove the contrapositive of the statement. Let the algorithm outputs \accept, then we need to prove that $D$ is $7\epsilon$ close to monotone.

As the algorithm accepts, $\sum_{I_j\in J}\Tilde{D}(I_j)\leq 5\epsilon$, for the set of intervals $J$ where 
$|S_1||S_2|\geq O(|S_{I_j}|/\epsilon^4)$ and $\frac{coll(S_1,S_2)}{|S_1||S_2|}\geq(\frac{1+\epsilon^2/64}{|I_j|}+\frac{\epsilon^2}{16})$. For all such intervals $I_j\in J$ by Lemma \ref{basic_0}, we obtain $d_{TV}(D_{I_j},\mathcal{U}_{I_j}) \geq \frac{\epsilon}{4} $.

Now, we calculate the distance between $D$ and the flattened distribution and we get $ d_{TV}(D,(D^f)^{\mathcal{I}})<4\epsilon $ 

We also know from Lemma \ref{lem:folklore-learn}, $d_{TV}((D^f)^{\mathcal{I}},(\Tilde{D}^f)^{\mathcal{I}})<\epsilon$. By triangle inequality, $d_{TV}(D,(\Tilde{D}^f)^{\mathcal{I}}) <5\epsilon $. As the algorithm outputs accept, there exists a monotone distribution $M$, such that $d_{TV}(\Tilde{D}^f)^{\mathcal{I}}, M)\leq 2\epsilon$. By triangle inequality, we have $d_{TV}(D,M)<7\epsilon$.
\end{proof}
\subsection{Testing Monotonicity in the streaming model}
In this section, we present the monotonicity tester in the streaming settings. A set of samples is drawn according to the standard access model that is revealed online one at a time. The task is to test whether an unknown distribution is a monotone or $\epsilon$ far from monotonicity. Also, there is a memory bound of $m$ bits. We use the notion of bipartite collision monotonicity tester \ref{bipartite_tester} discussed in the previous section. For satisfying the memory bound, we store an optimal number of samples for such intervals and count bipartite collision between the stored samples and the remaining ones. We present the algorithm below,
\begin{algorithm}
\caption{Streaming Monotonicity}
\SetKwInOut{Input}{Input}
        \SetKwInOut{Output}{Output}
\Input{\samp access to $D$, $\ell=O(\frac{1}{\epsilon_1}\log{(n\epsilon_1+1)})$ oblivious partitions $\mathcal{I}=\{I_1,..,I_{\ell}\}$ and error parameter $\epsilon,\epsilon_1\in (0,1]$, where $\epsilon_1=\epsilon^2$, memory requirement $\log^2{n}/\epsilon^6\leq m\leq \sqrt{n}/\epsilon^3$}
Sample $T=\tilde{O}(\frac{1}{\epsilon^6}\log^2{n})$ points from \samp\\
Get the empirical distribution $\Tilde{D}$ over $\ell$\\
Obtain an additional sample $S=O(\frac{n\log{n}}{m\epsilon^{8}})$ from \samp\\
For each interval store the first set of $S_1=O(\frac{m\epsilon^2}{\log^2{n}})$ samples in memory\\
Let $J$ be the set of intervals, where for the next set of $S_2=O(\frac{n}{m\epsilon^{4}})$ points, the following condition is satisfied,
$\frac{coll(S_1,S_2)}{|S_1||S_2|} \geq(\frac{1+\epsilon^2/64}{|I_j|}+\frac{\epsilon^2}{16})  $\\
Check \If{$\sum_{I_j\in J} \Tilde{D}(I_j)>5\epsilon$}{\reject and Exit}\label{line_61}
Define a flat distribution $(\Tilde{D}^f)^{\mathcal{I}}$ over $[n]$\\
Output \accept if $(\Tilde{D}^f)^{\mathcal{I}}$ is $2\epsilon$-close to a monotone distribution. Otherwise output \reject\label{algo_line21}
\end{algorithm}
\begin{theorem}
The algorithm {\sc streaming monotonicity} uses $O(\frac{n\log{n}}{m\epsilon^{8}})$ \samp queries and outputs \accept with probability at least $2/3$ if $D$ is a monotone distribution and outputs \reject with probability at least $2/3$ when $D$ is not $7\epsilon$ close to monotone. It uses $O(m)$ bits of memory for $\log^2{n}/\epsilon^6\leq m\leq \sqrt{n}/\epsilon^3$.
\end{theorem}

\begin{proof}
As there are $O(\frac{\log{n}}{\epsilon^2})$ partitions, there will be at least one interval with $D(I_j)\geq \frac{\epsilon^2}{\log{n}} $. An application of Chernoff bound shows that with high probability all such intervals contain $|S_{I_j}|=O(n/m\epsilon^4)$ points. In the algorithm, we divide $S_{I_j}$ into two sets $S_1$ and $S_2$ such that for $\log^2{n}/\epsilon^6\leq m\leq \sqrt{n}/\epsilon^3$, $|S_1|+|S_2|=O(m\epsilon^2/\log^2{n})+O(n/m\epsilon^4)=O(n/m\epsilon^4)$ and $|S_1|.|S_2|= O(n/\epsilon^2\log^2{n})\geq O(n/m\epsilon^8)=(1/\epsilon^4)|S_{I_j}|$. (The inequality is obtained by the fact that $m\geq \log^2{n}/\epsilon^6$). This implies that the condition of Lemma \ref{basic_0} is satisfied by these intervals and they are eligible for estimating the collision probability using bipartite collision count. The rest of the analysis follows from Theorem \ref{bipartite}.

The algorithm uses $O(m)$ bits of memory for implementation in a single-pass streaming model. For obtaining the empirical distribution $\Tilde{D}$, we will use one counter for each of the $\ell$ intervals. When a sample $x$ comes, if $x\in I_j$, the corresponding counter for $I_j$ will be incremented by $1$. In the end, the counters will give the number of samples that fall in each of the intervals, and using those values we can explicitly obtain the distribution $\Tilde{D}$. Each counter takes $O(\log{n})$ bits of memory. There are total $\ell=(\log{n}/\epsilon^2)$ counters. So, the memory requirement for this step is $O(\log^2{n}/\epsilon^2)<m$ bits. Also, using the distribution $\tilde{D}$ we can obtain the flattened distribution $(\Tilde{D}^f)^{\mathcal{I}}$ without storing it explicitly. Hence, the Line \ref{algo_line21} does not require any extra space for checking whether $(\Tilde{D}^f)^{\mathcal{I}}$ is $2\epsilon$ close to monotone or not.
For storing the first set of $S_1=O(m\epsilon^2/\log^2{n})$ samples for an interval will take $O(m\epsilon^2/\log{n})$ bits of memory. As we are storing $S_1$ samples for all $\ell=O(\log{n}/\epsilon^2)$ intervals, it will take total $O(m)$ bits of memory. 
\end{proof}
\textbf{Remark}
If the input to the algorithm is a monotone distribution, then the streaming algorithm computes a distribution over the intervals $\mathcal{I}$ such that the flattening is close to a monotone distribution. Since the number of intervals in the partition is $O(\log n/\epsilon)$, the explicit description of the distribution can be succinctly stored.

We would also like to point out that the final step in the algorithm requires testing if the learnt distribution is close to some monotone distribution, and we have not explicitly bounded the space required for that.
\subsubsection{Lower bound for testing monotonicity} \label{subsec:lower-bound}
In this section, we prove the lower bound for monotonicity testing problem in the streaming settings. We start with the discussion of the uniformity testing lower bound by (\cite{DGKR19}) in the streaming model and later we show how the same lower bound is applicable in our case.
\begin{theorem}[Uniformity testing lower bound in streaming framework \cite{DGKR19}]
Let $\mathcal{A}$ be an algorithm which tests if a distribution $D$ is uniform versus $\epsilon$-far from uniform with error probability $1/3$, can access the samples in a single-pass streaming fashion using $m$ bits of memory and $S$ samples, then $S.m=\Omega(n/\epsilon^2)$. Furthermore, if $S<n^{0.9}$ and $m>S^2/n^{0.9}$ then $S\cdot m=\Omega(n\log{n}/\epsilon^4)$.
\label{thm:uni-lowerboud}
\end{theorem}
The proof of the above lemma proceeds by choosing a random bit $X\in\{0,1\}$, where $X=0$ defines a \emph{Yes} instance (uniform distribution) and $X=1$ defines a \emph{No} instance ($\epsilon$-far from uniform) and calculating the mutual information between $X$ and the bits stored in the memory after seeing $S$ samples. In their formulation, the \emph{Yes} instance is a uniform distribution over $2n$ and the \emph{No} instance is obtained by pairing $(2i-1,2i)$ indices together and assigning values by tossing an $\epsilon$-biased coin. In particular, the \emph{No} distribution is obtained as follows, pair the indices as $\{1,2\}, \{3,4\},...,\{2n-1,2n\}$. Pick a bin $\{2i-1,2i\}$ and for each bin a random bit $Y_i\in \{\pm 1\}$ to assign the probabilities as,

$$
(D(2i-1),D(2i))=
    \begin{cases}
      \frac{1+\epsilon}{2n}, \frac{1-\epsilon}{2n}& \text{if}\ Y_i=1 \\
           \frac{1-\epsilon}{2n}, \frac{1+\epsilon}{2n}& \text{if}\ Y_i=-1 
    \end{cases}
$$

It is straightforward that the \emph{Yes} distribution is a monotone distribution as well. We show that any distribution $D$ from the \emph{No} instance set is $O(\epsilon)$-far from monotonicity. We start by choosing an $\alpha\in(0,\epsilon/4)$ and defining a set of partitions $\mathcal{I}=\{I_1,...,I_{\ell}\}$ such that $|I_j|=\lfloor (1+\alpha)^j\rfloor$ for $1\leq j\leq \ell$. Let $(D^f)^{\mathcal{I}}$ be the flattened distribution corresponding to $\mathcal{I}$. We use the following lemma from (\cite{Can15}) which reflects the fact if $D$ is far from $(D^f)^{\mathcal{I}}$, then $D$ is also far from being monotone. In particular, we define the lemma as follows,
\begin{lemma}[\cite{Can15}] \label{lem:mon-far} Let $D$ be a distribution over domain $[n]$ and $\mathcal{I}=\{I_1,...,I_{\ell}\}$ are the set of partitions defined obliviously with respect to a parameter $\alpha\in (0,1)$ where $\ell=O(\frac{1}{\alpha}\log{n\alpha}) $ and $|I_j|=\lfloor (1+\alpha)^j\rfloor$. If $D$ is $\epsilon$-close to monotone non-increasing, then $d_{TV}(D,(D^f)^{\mathcal{I}})\leq 2\epsilon+\alpha$ where $(D^f)^{\mathcal{I}}$ is the flattened distribution of $D$ with respect to $\mathcal{I}$.
\end{lemma}

Let, $D$ be a distribution chosen randomly from the \emph{No} instance set. We have the following observation,
\begin{lemma} \label{lem:odd-interval}
    Let $\mathcal{I}=\{I_1,...,I_{\ell}\}$ be the oblivious partitions of $D$ with parameter $\alpha$ such that $|I_j|=\lfloor (1+\alpha)^j\rfloor$. 
    \begin{itemize}
    \item If $|I_j|$ is odd, then $\sum_{i\in I_j}|D(i)-\frac{D(I_j)}{|I_j|}|=\frac{\epsilon}{2n}(|I_j|-\frac{1}{|I_j|})$.
    \item If $|I_j|$ is even, then $\sum_{i\in I_j}|D(i)-\frac{D(I_j)}{|I_j|}|\geq\frac{\epsilon}{2n}(|I_j|-\frac{4}{|I_j|})$.
    \end{itemize}
\end{lemma}
\begin{proof}
    If $|I_j|$ is odd, it will contain $k$ (any positive integer) number of bin where each bin is of form $(2x-1,2x)$ and an extra index $i'$ which can have the probability weight either $\frac{1+\epsilon}{2n}$ or $\frac{1-\epsilon}{2n}$. Let $D(i')=\frac{1+\epsilon}{2n}$. In this case, $D(I_j)=\frac{|I_j|}{2n}+\frac{\epsilon}{2n}$. 
\begin{align*}
    \sum_{i\in I_j}|D(i)-\frac{D(I_j)}{|I_j|}|&=\sum_{i\in I_j}|D(i)-\frac{1}{2n}-\frac{\epsilon}{2n|I_j|}|\\
    &=\frac{\epsilon}{2n}(1-\frac{1}{|I_j|})\frac{|I_j|-1}{2}+\frac{\epsilon}{2n}(1+\frac{1}{|I_j|})\frac{|I_j|-1}{2}+\frac{\epsilon}{2n}(1-\frac{1}{|I_j|})\\
    &=\frac{\epsilon}{2n}(|I_j|-\frac{1}{|I_j|})
\end{align*}
    When $D(i')=\frac{1-\epsilon}{2n}$, similar calculation will follow.

    If $|I_j|$ is even, there are two possibilities, $(i)$ $I_j$ consists of $k$ (positive integer) bins. So, there will be equal number of $\frac{1+\epsilon}{2n}$ and $\frac{1-\epsilon}{2n}$ in $I_j$ and $D(I_j)=\frac{|I_j|}{2n}$. In this case, it is straightforward to observe that  $\sum_{i\in I_j} |D(i)-\frac{D(I_j)}{|I_j|}|=\frac{\epsilon |I_j|}{2n}$. Another case is, $(ii)$ $I_j$ contains $b_p,...,b_{p+k-1}$ bins completely and $i'\in b_{p-1}$, and $i''\in b_{p+k}$ where $D(i')=D(i'')$; the case when $D(i')\neq D(i'')$ will be similar to $(i)$ that we saw earlier. Let $D(i')=D(i'')=\frac{1+\epsilon}{2n}$. In this case, $D(I_j)=\frac{|I_j|}{2n}+\frac{\epsilon}{n}$. 

\begin{align*}
    \sum_{i\in I_j}|D(i)-\frac{D(I_j)}{|I_j|}|&=\sum_{i\in I_j}|D(i)-\frac{1}{2n}-\frac{\epsilon}{n|I_j|}|\\
    &=\frac{\epsilon}{2n}(1-\frac{1}{|I_j|})\frac{|I_j|-2}{2}+\frac{\epsilon}{2n}(1+\frac{1}{|I_j|})\frac{|I_j|-2}{2}+\frac{\epsilon}{n}(1-\frac{2}{|I_j|})\\
    &=\frac{\epsilon}{2n}(|I_j|-\frac{4}{|I_j|})
\end{align*}
Combining $(i)$ and $(ii)$, we say $\sum_{i\in I_j}|D(i)-\frac{D(I_j)}{|I_j|}|\geq\frac{\epsilon}{2n}(|I_j|-\frac{4}{|I_j|})$. Similar calculation will follow when $D(i')=D(i'')=\frac{1-\epsilon}{2n}$.

\end{proof}

In our case, we apply oblivious partitions on $D$ (chosen randomly from the \emph{No} set) with respect to the parameter $\alpha$ and conclude the following,
\begin{lemma}
    Let $D$ be a randomly chosen distribution from the No instance set, then $D$ is $\epsilon/4$-far from any monotone non-increasing distribution.
\end{lemma}
\begin{proof}
    We calculate $d_{TV}(D,(D^f)^{\mathcal{I}})=\sum_{j=1}^{\ell}\sum_{i\in I_j}|D(i)-\frac{D(I_j)}{|I_j|}|=\sum_{|I_j|\text{is even}}\sum_{i\in I_j}|D(i)-\frac{D(I_j)}{|I_j|}|+ \sum_{|I_j| \text{is odd}}\sum_{i\in I_j}|D(i)-\frac{D(I_j)}{|I_j|}|$. Each odd length interval contributes $\sum_{i\in I_j}|D(i)-\frac{D(I_j)}{|I_j|}|=\frac{\epsilon}{2n}(|I_j|-\frac{1}{|I_j|})$ and each even length interval contributes $\sum_{i\in I_j}|D(i)-\frac{D(I_j)}{|I_j|}|\geq\frac{\epsilon}{2n}(|I_j|-\frac{4}{|I_j|})$ by using Lemma \ref{lem:odd-interval}. 
    
 Hence, simplifying the distance, we get,
    $d_{TV}(D,(D^f)^{\mathcal{I}})\geq\sum_{|I_j| \text{is even}}\frac{\epsilon}{2n}(|I_j|-\frac{4}{|I_j|})+\sum_{|I_j| \text{is odd}}\frac{\epsilon}{2n}(|I_j|-\frac{1}{|I_j|})\geq\frac{\epsilon}{2n}\sum_{I_j\in \ell} |I_j|-\frac{\epsilon}{2n}\big(\sum_{|I_j| \text{is even}}\frac{4}{|I_j|}+\sum_{|I_j| \text{is odd}}\frac{1}{|I_j|}\big)\geq \epsilon-\frac{\epsilon}{2n}.5\ell\geq\frac{3\epsilon}{4}>2\frac{\epsilon}{4}+\alpha$. The third inequality is obtained by using the fact that $|I_j|\geq 1$ and the fourth inequality considers $\ell<n/10$. Now, by using the contra-positive of the Lemma \ref{lem:mon-far}, $D$ is $\epsilon/4$-far from any monotone non-increasing distribution.
\end{proof}
Therefore, the uniformity testing lower bound from \cite{DGKR19} is applicable in our case for distinguishing monotone from $\epsilon/4$-far monotone. We formalize this in the theorem below.
\begin{theorem}
    Let $\mathcal{A}$ be an algorithm that tests if a distribution $D$ is monotone versus $\epsilon/4$-far from monotonicity with error probability $1/3$, can access the samples in a single-pass streaming fashion using $m$ bits of memory and $S$ samples, then $S.m=\Omega(n/\epsilon^2)$. Furthermore, if $n^{0.34}/\epsilon^{8/3}+n^{0.1}/\epsilon^4\leq m\leq \sqrt{n}/\epsilon^3$, then $S.m=\Omega(n\log{n}/\epsilon^4)$.
\end{theorem}
We obtain the above theorem as analogous to the Theorem \ref{thm:uni-lowerboud} by showing that lower bound for uniformity implies lower bound for monotonicity in the streaming framework. In particular, the uniform distribution is monotone non-increasing by default and we show that a randomly chosen distribution from \emph{No} instance set is $\epsilon/4$-far from monotone no-increasing. Hence, the correctness of the above theorem follows directly from the Theorem \ref{thm:uni-lowerboud}.

\section{Learning decomposable distributions in the streaming model}
\label{sec:gamma-l}
The algorithm and analysis from the previous section of monotone distributions extend to a more general class of structured distributions known as $(\gamma, L)$-decomposable distributions. The class of $(\gamma,L)$-decomposable distributions were first studied by Canonne et al (\cite{CanonneDGR18}), who gave a unified algorithm for testing monotonicity, k-modal, histograms, log-concave distributions since $(\gamma,L)$-decomposable distributions contain these other classes. We will first recall the definition.

\begin{definition}[$(\gamma, L)$-decomposable distribution \cite{CanonneDGR18}]
A class $\mathcal{C}$ of distributions is said to be $(\gamma, L)$-decomposable, if for every $D \in \mathcal{C}$, there exists an $\ell \leq L$ and a partition $\mathcal{I}=\{I_1,..,I_{\ell}\}$ of $[n]$ into intervals such that for every interval $I_j \in \mathcal{I}$ one of the following conditions hold.
\begin{itemize}
    \item $D(I_j)\leq \frac{\gamma}{L}$
    \item $max_{i\in I_j} D(i)\leq (1+\gamma)min_{i\in I_j} D(i)$
\end{itemize}
\end{definition}

The following lemma shows that monotone distributions, in particular, are decomposable.
\begin{lemma}[\cite{CanonneDGR18}]
For all $\gamma>0$, the class of monotone distributions $\mathcal{M}$ over $[n]$ is $(\gamma,L)$-decomposable, where $L=O(\frac{\log^2{n}}{\gamma})$.
\end{lemma}

To obtain an algorithm with trade-offs between sample complexity and space complexity, we will start with the algorithm of Fischer et al (\cite{FLV19}) that improves the sample complexity of \cite{DGKR19}. We will describe an algorithm that obtains an explicit description of an unknown $(\gamma, L)$-decomposable distribution. To that end, we start with the definition of an $(\eta,\gamma)$-fine partition as defined in \cite{FLV19}.

\begin{definition}[$(\eta, \gamma)$-fine Partition]
 Let $D$ be distribution over $[n]$ and $\mathcal{I}=\{I_1,...,I_{\ell}\}$ be an interval partition of $D$. $\mathcal{I}$ is said to be $(\eta, \gamma)$ fine partition if there exists $\eta>0$, $\gamma>0$ and a set $H\subset \mathcal{I}$, such that $H=\{I_j\in \mathcal{I} : D(I_j)>\eta,|I_j|>1\}$ and $\sum_{I_j\in H}D(I_j)\leq \gamma$.
\end{definition}
A set of $(\eta,\gamma)$ partitions can be obtained in the following way: sample $k=O(\frac{1}{\eta}\log{1/\gamma^{\delta}})$ points from $D$, sort them in increasing order $\{x_1<x_2<...<x_k\}$ without repetition and set $x_0=0$. For every point $x_j$; $1\leq j\leq k$ a singleton interval is added and for $x_j>x_{j-1}+1$, an interval $[x_{j-1}+1,x_{j}-1]$ is added. Finally for $x_k<n$, an interval $[x_k,n]$ is also added. Precisely, the following theorem can be summarised:
\begin{theorem}[\cite{FLV19}]\label{thm_pull}
Let $D$ be a distribution over $[n]$. For the parameters $\eta>0,\gamma>0, \delta>0$, there exists an algorithm that uses $O(\frac{1}{\eta}\log{1/\gamma^{\delta}})$ \samp queries and with probability at least $(1-\delta)$, finds a set of $(\eta,\gamma)$ fine partitions $\mathcal{I}=\{I_1,...,I_r\}$ of $D$ where $r=|\mathcal{I}|=O(\frac{1}{\eta}\log{1/\gamma^{\delta}})$.
\end{theorem}
After a set of $(\eta,\gamma)$ partitions is obtained, a \emph{weakly tolerant interval uniformity tester} is used to check how many of the intervals are far from uniformity. If a significant number of intervals are far from uniformity, then the obtained partitions can not be used for learning. Otherwise, the partitions are used to construct a distribution according to \ref{lem:folklore-learn}. The following theorem reflects the task of a \emph{weakly tolerant interval uniformity tester}:
\begin{theorem}[\cite{Pan08}]
  Let $D$ be a distribution over $[n]$. There exists an algorithm $\mathcal{A}$ which takes the following as inputs: \samp access to a distribution $D$, an interval $I\subset [n]$, a parameter $m$ defined as the maximum size of an interval, error parameters $0< \epsilon\leq 1$, $0\leq \delta\leq 1$. The algorithm does the following:
\begin{itemize}
    \item If $|I|\leq m$, $D(I)\geq \gamma$ and $bias(D\upharpoonright{I})\leq \frac{\epsilon}{100}$, then the algorithm accepts with probability at least $1-\delta$.
    \item If $|I|\leq m$, $D(I)\geq \gamma$, and $d_{TV}(D\upharpoonright{I},U_{I})>\epsilon$, then the algorithm rejects with probability at least $1-\delta$.
\end{itemize}
In all other cases, the algorithm behaves arbitrarily. The algorithm requires $O(\sqrt{m}\log{(1/\delta)}/\gamma \epsilon^2)$ samples from $D$.
\end{theorem}
Now, we explain how we implement the above-mentioned ideas in the streaming settings. It is easy to observe that the Theorem \ref{thm_pull} can be used as it is in the streaming settings without storing any samples. Consider a set of $O(\frac{1}{\eta}\log{1/\gamma^{\delta}})$ points that appear online as a stream one at a time. We can construct the singleton intervals online by looking at the sampled points. Later, we can add the intervals lying in between two singleton intervals.
However, the weakly tolerant interval uniformity tester requires $O(\sqrt{m}\log{(1/\delta)}/\gamma \epsilon^2)$ samples (\cite{Pan08}) for each $I_j$ where $|I_j|\leq m$. This leads to the use of excessive memory storage. We observe that the weakly tolerant interval uniformity tester's task can be replaced by the use of bipartite collision count and we present the following algorithm, which is a small modification of \cite{FLV19}.

\begin{algorithm}[H]
   \caption{Assessing a Partition Streaming} \label{algo_assesing_partition}
\SetKwInOut{Input}{Input}
        \SetKwInOut{Output}{Output}
\Input{\samp access to $D$, a $(\eta,\gamma)$-fine interval partitions $\mathcal{I}=\{I_1,...,I_r\}$, parameters $c, r$, error parameters $\epsilon\in(0,1)$ memory requirement $\log{n}/\epsilon^4 \leq m\leq \sqrt{n\log{n}}/\epsilon^3$}
Sample $T=O(\frac{1}{\epsilon^4}r^2\log{r})$ points from $SAMP_D$\\
Estimate the frequencies of each interval $I_j$ to be $\tilde{f}_{I_j}$ (by adding the frequencies of all elements in $I_j$); using \countmin sketch $(\epsilon,\delta)$\\
\For{$k=O(1/\epsilon)$ times}
{
Obtain an additional sample $S=O(\frac{nr}{m\epsilon^8})$ from $SAMP_D$\\
Let $x$ be a sample and $I_j\in \mathcal{I}$ be the interval that contains it\\
\If{$|I_j|\leq \frac{n}{c}$ and $\frac{\tilde{f}_{I_j}}{|T|}\geq \frac{\epsilon}{r}-\frac{\epsilon^2}{r}$} 
{
Store first set of $S_1=O(\frac{m}{\log{n}})$ samples in memory\\
For next set of $S_2=O(\frac{n}{m\epsilon^5})$ points in $I_j$ check\\
\lIf{$\frac{coll(S_1,S_2)}{|S_1||S_2|} >\frac{1+63\epsilon^2/64}{|I_j|}$}
{Add $I_j$ to $\mathcal{B}$}
}
}
\leIf{$|\mathcal{B}|> 4\epsilon k$}{\reject}{\accept}
\end{algorithm}
The following theorem shows how bipartite collision count does the same task as that of the weakly tolerant interval uniformity tester.
\begin{theorem}\label{partition}
    Let $D$ be a distribution over $[n]$, $I\subset [n]$ be an interval such that $D(I)\geq \epsilon/r$. Let $S_I$ be the set of samples that falls inside $I$ while sampling $S=O(\frac{nr}{m\epsilon^8})$ points according to $D$. Consider $S_I$ can be divided into two sets $S_1$ and $S_2$ such that $|S_1|\cdot|S_2|\geq |S_I|/\epsilon^4$. Then the following happens with high probability,
    
    \begin{itemize}
        \item If $bias(D\upharpoonright{I})\leq \frac{\epsilon}{100}$, then $\frac{coll(S_1,S_2)}{|S_1||S_2|} \leq \frac{1+\epsilon'+\epsilon^2/64}{|I_j|}$; where $\epsilon'=\frac{\epsilon^2}{10^4}$
    \item If $d_{TV}(D\upharpoonright{I},U_{I})>\epsilon$, then $\frac{coll(S_1,S_2)}{|S_1||S_2|} > \frac{1+63\epsilon^2/64}{|I|}$
    \end{itemize}
\end{theorem}

\begin{proof}
   As $D(I)\geq \epsilon/r$, an additive Chernoff bound shows that $I$ will contain $|S_I|=O(\frac{n}{m\epsilon^5})$ points out of $S=O(\frac{nr}{m\epsilon^8})$ samples with probability at least $2/3$. Also considering $S_1=O(\frac{m}{\log{n}})$ and  $S_2=O(\frac{n}{m\epsilon^5})$, we get $|S_1|\cdot|S_2|=O(n/\epsilon^5\log{n})\geq O(\frac{n}{m\epsilon^9})$ for $m\geq \log{n}/\epsilon^4$. This implies that  $|S_1|\cdot|S_2|\geq \frac{|S_I|}{\epsilon^4}$. Hence, by Lemma \ref{basic_0} the following happens with probability at least $2/3$, $||D_{I}||_2^2-\frac{\epsilon^2}{64|I|}\leq \frac{coll(S_1,S_2)}{|S_1||S_2|}\leq ||D_{I}||_2^2+\frac{\epsilon^2}{64|I|}$.
Let $bias(D\upharpoonright{I})\leq \frac{\epsilon}{100}$ which implies that for $x\in I$, $max_{x\in I}D(x)\leq (1+\epsilon/100) min_{x\in I}D(x)$. Using the Lemma \ref{observation} we get, $||D_{I}||_2^2\leq \frac{1+\epsilon^2/10^4}{|I|}$. Hence, we obtain, $\frac{coll(S_1,S_2)}{|S_1||S_2|} < \frac{1+\epsilon'+\epsilon^2/64}{|I|}$.
Let $d_{TV}(D\upharpoonright{I},U_{I})>\epsilon$, by the Lemma \ref{observation}, $||D_{I}||_2^2> \frac{1+\epsilon^2}{|I|}$. As  $D(I)\geq \epsilon/r$, we already proved $|S_1|\cdot|S_2|\geq \frac{|S_I|}{\epsilon^4}$. Applying Lemma \ref{basic_0}, we get $\frac{coll(S_1,S_2)}{|S_1||S_2|} >  \frac{1+63\epsilon^2/64}{|I|}$.
\end{proof}
The above theorem shows that the acceptance and rejection conditions of the weakly tolerant interval uniformity tester can be substituted by the bipartite collision count except from the fact that $|I_j|\leq n/c$ is not examined. We check this just by adding an extra condition in our algorithm.

\begin{theorem}
The algorithm {\sc assessing a partition streaming} takes a set of $(\eta,\gamma)$-fine interval partitions $\mathcal{I}=\{I_1,...,I_r\}$ as input, uses $S=O(\frac{nr}{m\epsilon^8})$ samples according to the standard access oracle and does the following when $c\eta+\gamma\leq\epsilon$,
    \begin{itemize}
        \item Define $\mathcal{G_I}=\{I_j \in \mathcal{I}: bias(D\upharpoonright{I_j})\leq \frac{\epsilon}{100}\}$. If $D(\cup_{I_j\in \mathcal{G_I}})\geq 1-\epsilon$, then the algorithm outputs Accept with probability at least $2/3$.
        \item Define $\mathcal{F_I}=\{I_j\in \mathcal{I} : d_{TV}(D\upharpoonright{I_j},U_{I_j})>\epsilon\}$. If $D(\cup_{I_j\in \mathcal{F_I}})\geq 7\epsilon$, then the algorithm outputs Reject with probability at least $2/3$.
    \end{itemize}
    The memory requirement for the algorithm is $O(m)$ bits, where $\log{n}/\epsilon^4 \leq m\leq O(\sqrt{n\log{n}}/\epsilon^3)$.
\end{theorem}
The correctness of the above theorem follows from \cite{FLV19}, we give a brief outline as follows,
\begin{proof}
  Let us define the set $\mathcal{N_I}=\{I_j:|I_j|>n/c \ or, D(I_j)<\epsilon/r\}$. It is observed that for a set of $(\eta,\gamma)$-fine intervals where $c\eta+\gamma\leq\epsilon$, $D(\mathcal{N_I})\geq 2\epsilon$. Hence, $D(\mathcal{G_I}\setminus \mathcal{N_I})\geq (1-3\epsilon)$. As a result, out of $O(1/\epsilon)$ iterations, at most $4\epsilon k$ intervals are drawn from the desired set where $|I_j|\leq \frac{n}{c}$ and $D(I_j)\geq \frac{\epsilon}{r}$ and $bias(D\upharpoonright{I_j})\leq \frac{\epsilon}{100}$. Furthermore, these intervals can be caught by counting the number of bipartite collisions and they are indeed correct by Theorem \ref{partition}. Hence, the algorithm outputs \accept.

By the definition of $\mathcal{F_I}$, it is easy to observe that $D(\mathcal{F_I}\setminus \mathcal{N_I})\geq 5\epsilon$. As a result, out of $O(1/\epsilon)$ iterations, more than $4\epsilon k$ intervals are drawn from the desired set where $|I_j|\leq \frac{n}{c}$ and $D(I_j)\geq \frac{\epsilon}{r}$ and $ d_{TV}(D\upharpoonright{I_j},U_{I_j})>\epsilon$. Furthermore, these intervals will be caught by Theorem \ref{partition} and the algorithm outputs \reject with high probability.

\textbf{Space complexity :} The first set of points $T=O(\frac{1}{\epsilon^4}r^2\log{r})$ is sampled for estimating weights of each interval $D(I_j)$. An additive Chernoff bound (followed by a union bound over $r$) shows that by using $T$ samples, for all intervals $I_j$, $|D(I_j)-\tilde{D}(I_j)|\leq\epsilon^2/r$. Instead of storing all the samples, we use \countmin sketch with parameters $(\epsilon,\delta)$ to save the space. When the samples appear one at a time as stream of $T$ elements, we store the frequencies of each element in the \countmin table. If $f_{x}$ be the frequency of an element $x\in T$, by Lemma \ref{cmin}, with probability at least $(1-\delta)$, $f_{x}\leq \tilde{f}_{x}\leq f_{x}+\epsilon |T|$. We can get the frequency of an interval $I_j$ by adding the frequencies of all the elements lying in $I_j$, i.e $f_{I_j}=\sum_{x\in I_j}f_{x}$. We observe that $f_{I_j}\leq \tilde{f}_{I_j}\leq f_{I_j}+ \epsilon |T|^2$. We also know that $\tilde{D}(I_j)=\frac{f_{I_j}}{|T|}$ and $\tilde{D}(I_j)\geq D(I_j)-\epsilon^2/r$. By combining these, to check if $D(I_j)\geq \epsilon/r$, it would be sufficient to check if $\frac{f_{I_j}}{|T|}\geq \epsilon/r-\epsilon^2/r$. The space used for this procedure is $O(\epsilon \log{1/\delta})<m$ by the use of \countmin $(\epsilon,\delta)$. For the rest of the algorithm, we are storing the set $|S_1|=O(\frac{m}{\log{n}})$ samples in memory. So, a total of $|S_1|\cdot\log{n}=O(m)$ bits storage is required for the implementation of the algorithm.
\end{proof}

Now, we describe the final learning algorithm for $(\gamma,L)$-decomposable distribution in the one-pass streaming settings.

\begin{algorithm}[H]
     \caption{Learning $L$-decomposable Distribution Streaming}
\SetKwInOut{Input}{Input}
        \SetKwInOut{Output}{Output}
\Input{\samp access to $D$ supported over $[n]$, parameters $c=20, r=10^5L\log{(1/\epsilon)}/\epsilon$, error parameters $\epsilon,\delta$, memory requirement $\log{n}/\epsilon^4 \leq m\leq O(\sqrt{n\log{n}}/\epsilon^3)$}
\Output{An explicit distribution $(\Tilde{D}^f)^{\mathcal{I}}$}
Use Theorem \ref{thm_pull} to obtain a set of $(\epsilon/2000L,\epsilon/2000)$ fine partitions of $[n]$\\
Run algorithm {\sc assessing partition streaming} \\
\eIf{it Rejects}{\reject}{
Return the flattened distribution of $\Tilde{D}$, i.e., $(\Tilde{D}^f)^{\mathcal{I}}$
}
\label{alg:learn-gamma-l}
\end{algorithm}
The correctness of the algorithm follows from Lemma $7.1$ from (\cite{FLV19}). Our adaption of the lemma is as follows:
\begin{theorem}
If $D$ is an $(\epsilon/2000,L)$-decomposable distribution, then the algorithm {\sc learning $L$-decomposable distribution streaming} outputs a distribution $(\Tilde{D}^f)^{\mathcal{I}}$ such that $d_{TV}(D,(\Tilde{D}^f)^{\mathcal{I}})\leq \epsilon$ with probability at least $1-\delta$. The algorithm requires $O(\frac{nL\log{(1/\epsilon)}}{m\epsilon^9})$ samples from $D$ and needs $O(m)$ bits of memory where $\log{n}/\epsilon^4 \leq m\leq O(\sqrt{n\log{n}}/\epsilon^3)$.
\end{theorem}

The above algorithm can be used as a subroutine for testing $(\gamma,L)$-decomposable properties. Given an unknown distribution $D$, we will use Algorithm~\ref{alg:learn-gamma-l} to learn an explicit description of the distribution and the test if it is $(\gamma,L)$-decomposable. Once again, like in the monotone distribution case, we note that the final testing of the explicit description will require additional space that we haven't accounted in the earlier algorithm.

\begin{theorem}
    Let $\mathcal{C}$ be a $(\gamma,L)$-decomposable property for $L=L(\epsilon/4000,n)$. The algorithm {\sc testing $L$-decomposable properties streaming} requires $O(\frac{nL\log{(1/\epsilon)}}{m\epsilon^9})$ samples from $D$ and does the following,
    \begin{itemize}
        \item If $D$ satisfies $\mathcal{C}$, it outputs Accept with probability at least $(1-\delta)$
        \item If $D$ is $2\epsilon$ far form $\mathcal{C}$, it outputs Reject with probability $(1-\delta)$
    \end{itemize}
   The algorithm uses $O(m)$ bits of memory where $\log{n}/\epsilon^4 \leq m\leq O(\sqrt{n\log{n}}/\epsilon^3)$.
\end{theorem}
\begin{proof}
    Let $D$ satisfies $\mathcal{C}$. The explicit distribution $(\Tilde{D}^f)^{\mathcal{I}}$ will be $\epsilon$ close to $\mathcal{C}$. Hence, the algorithm outputs Accept in this case. Similarly, when $D$ is $2\epsilon$ far from $\mathcal{C}$,  $(\Tilde{D}^f)^{\mathcal{I}}$ will be $\epsilon$ far from $\mathcal{C}$ and the algorithm outputs Reject. 

    The algorithm requires $O(L/\epsilon)$ samples from $D$ for finding the first set of $(\epsilon/2000L,\epsilon/2000)$ fine partitions by Theorem \ref{thm_pull}. However, an $O(\frac{nL\log{(1/\epsilon)}}{m\epsilon^9})$ samples required for performing the Algorithm \ref{algo_assesing_partition} which gives the total sample complexity. The memory requirement for the algorithm is $O(m)$ bits which is the same as required by the learning $L$- decomposable distribution in the streaming settings.
\end{proof}
\section{Conclusion}
We give efficient algorithms for testing identity, monotonicity and $(\gamma, L)$-decomposability in the streaming model. For a memory constraint $m$, the number of samples required is a function of the support size $n$ and the constraint $m$. For monotonicity testing, our bounds are nearly optimal. We note that the trade-off that we achieve, and lower bounds work for certain parameters of the value $m$. Furthermore, we have not tried to tighten the dependence of the bound on the parameter $\epsilon$. One natural question to ask is if the dependence of sample complexity on $m$ can be improved, and whether it can work for a larger range of values.

\newpage

\bibliographystyle{alpha}
\bibliography{paper}


\end{document}